\documentclass[reqno]{amsart}

\usepackage{amsmath,amssymb,amsthm,amsfonts}
\usepackage{color}
\usepackage{ulem}
\usepackage{hyperref}
\usepackage{a4wide}
\usepackage{mathrsfs} % for calligraphic letters \mathscr

\numberwithin{equation}{section}

\newtheorem{theorem}{Theorem}[section]
\newtheorem{lemma}[theorem]{Lemma}
\newtheorem{prop}[theorem] {Proposition}
\newtheorem{cor}[theorem]  {Corollary}
\newtheorem{definition}[theorem] {Definition}

\theoremstyle{definition}

\theoremstyle{remark}
\newtheorem{remark}[theorem]{Remark}
\newtheorem{example}[theorem]{Example}

\newcommand{\e}{\mathrm{e}} % exponential e
\newcommand{\N}{\mathbb{N}}
\newcommand{\R}{\mathbb{R}}
\newcommand{\Z}{\mathbb{Z}}
\newcommand{\C}{\mathbb{C}}

\newcommand{\dd}{\mathrm{d}} %integration d
\newcommand{\eps}{\varepsilon}
\newcommand{\la}{\langle}
\newcommand{\ra}{\rangle}

\newcommand{\vect}[1]{\boldsymbol{#1}}

\newcommand{\be}{\begin{equation}}
\newcommand{\ee}{\end{equation}}
\newcommand{\ba}{\begin{equation} \begin{aligned}}
\newcommand{\ea}{\end{aligned}\end{equation}}
\newcommand{\bes}{\begin{equation*}}
\newcommand{\ees}{\end{equation*}}

\def\1{{\mathchoice {1\mskip-4mu\mathrm l}      % Blackboard bold 1
{1\mskip-4mu\mathrm l}
{1\mskip-4.5mu\mathrm l} {1\mskip-5mu\mathrm l}}}
\newcommand{\esssup}{\operatornamewithlimits{ess\,sup}}

\begin{document}

%%%%%%%%%%%%%%
% Front matter %% 
\title{Virial inversion and density functionals}
\author{Sabine Jansen}
\address{Mathematisches Institut, Ludwig-Maximilians-Universit{\"a}t, 80333 M{\"u}nchen,  Germany}
\email{jansen@math.lmu.de} 
\author{Tobias Kuna}
\address{Department of Mathematics and Statistics,
University of Reading, Reading RG6 6AX, UK}
\email{t.kuna@reading.ac.uk}
\author{ Dimitrios Tsagkarogiannis}
\address{Dipartimento di Ingegneria e Scienze dell'Informazione e Matematica, Universit\`a degli Studi dell'Aquila, 67100 L'Aquila, Italy}
\email{dimitrios.tsagkarogiannis@univaq.it}

\date{30 August 2019}
%%%%%%%%
\begin{abstract} 
	We prove a novel inversion theorem for functionals given as power series in  infinite-dimensional spaces and apply it to the inversion of the density-activity relation for inhomogeneous systems. This provides a rigorous framework to prove
	convergence for density functionals for inhomogeneous systems with applications in classical density function theory, liquid crystals,
	molecules with various shapes or other internal degrees of freedom. The key technical tool is the representation of the inverse via a fixed point equation and a combinatorial identity for trees, which allows us to obtain convergence estimates in situations where Banach inversion fails. 
	Moreover, the new method for the inversion gives for the (homogeneous) hard sphere gas a significantly improved radius of convergence for the virial expansion improving the first and up to now best result  by Lebowitz and Penrose (1964).
\\

\noindent\emph{Keywords}: cluster and virial expansions -- density functional theory  -- holomorphic functions in Banach spaces \\
% colored combinatorial species %-- Lagrange-Good inversion.\\

\noindent\emph{MSC 2010 classification:}  82B05, 82D15, 82D30, 
%	46G20, %infinite-dimensional holomorphy
	47J07, %Abstract inverse mapping and implicit function theorems
	05C05 % trees
\end{abstract}

\maketitle

\tableofcontents
 %%%%%%%%%%%%% 

\section{Introduction}

Deriving functional expressions for thermodynamic quantities from microscopic models which are based on physical principles is one of the main challenges of both theoretical and computational methods in statistical mechanics. Furthermore, the use of such functionals is ubiquitous in applied mathematics for example in classical density function theory, liquid crystals, heterogenous materials, colloid systems, system of molecules with various shapes or other internal degrees of freedom. However, often the key point in variational calculus and the theory of PDE is to consider non constant densities and hence  non translation invariant systems. 
One key mathematically rigorous result in this direction was the proof of the convergence of the virial expansion by Lebowitz and Penrose in 1964~\cite{lebowitz-penrose1964virial},  building on the previously established convergence of the activity expansion of the pressure and of the density. 
The proof consists out of three main steps: first to invert the density-activity relation, second to plug the resulting expansion of the activity as a function of the density into the pressure-activity expansion and resum, and finally bound the radius of convergence of the composed power series combining convergence results for the inversions and for activity expansions. Previous results~\cite{mayerbook}, based on formal manipulations of power series and combinatorics of graphs, had already identified the coefficients in the density series in terms of two-connected (``irreducible'') graphs. A by-product of the convergence result from~\cite{lebowitz-penrose1964virial} is the absolute convergence of the generating function for two-connected graphs, thus justifying formulas that were already in use. 

This recipe for going from activity expansions to density expansions extends to  quantities whose activity expansion is well understood, for example, the truncated correlation functions. However convergence proofs for other quantities are more delicate, as explained in detail in \cite{kuna-tsagkaro2016} for the direct correlation functions. Indeed, even though combinatorial series for various quantities are available, their derivation rests on formal manipulations and graph re-summations that have yet to be rigorously justified.  
The formal graph re-summations were developed  in the 60's mainly by the works of Morita and Hiroike \cite{morita-hiroike1, morita-hiroike3} and of Stell  \cite{stell1964} on liquid state theory expansions for inhomogeneous fluids, allowing for position-dependent densities. In contrast, the convergence result from~\cite{lebowitz-penrose1964virial} and all subsequent works addresses homogeneous systems only. 

Our goal, therefore, is twofold: 
\begin{enumerate}
\item Establish the validity of the inversion formulas for inhomogeneous fluids.
\item Prove the validity of re-summation operations on graphs by showing that the resulting power series are absolutely convergent. 
\end{enumerate}

As far as  goal (2) is concerned,
in a previous work \cite{kuna-tsagkaro2016} we proved convergence for some resummation of expansions leading to graphs with higher connectivity properties, but starting from the canonical ensemble. That choice was made in order to avoid the graph re-summations that come with the inversion, but also since it is more natural for expansions with respect to the density. 
%Hence, establishing it in the grand-canonical case is a more delicate issue since it is more singular but also more
%relevant if we want to extend the proofs of convergence for expansions with respect to other quantities (than the density) such as the truncated correlation function.
In the current paper we prove the validity of these re-summations by inverting the density-activity relation, but a similar structure may be expected for other types of resummations as considered, e.g. inverting the truncated correlation vs activity relation.
We intend to
address all these issues in a subsequent work.

Concerning goal (1), since inhomogeneous system can be seen as a system of uncountably many species, when one considers the position $x\in \R^d$ as  species. We consider goal (1) in this more general context of a system of (potentially uncountably many) species. In this way, we can treat at the same time as well systems of mixtures as with internal degrees of freedom. This generalization will not increase the complexity of the arguments involved. 

At first sight, it may look as if goal (1) is best achieved with the help of inverse function theorems in complex Banach spaces, applied to the functional that maps the activity profile $(z(x))_{x\in \Lambda}$ to the density profile $(\rho(x))_{x\in \Lambda}$, see Section~\ref{sec:banach}. This works well for inhomogeneous systemsof e.g. objects of bounded size, e.g., hard spheres of fixed radius. It turns out, however, that Banach inversion fails for mixtures of objects of finite but unlimited size \cite{jttu2014,jansen2015tonks}, see Example~\ref{ex:unbounded}.  As a way out, mixtures of countably many species were treated with the help of Lagrange-Good inversion in~\cite{jttu2014}, leaving the case of uncountably many species wide open. 

Our first main result is a novel inversion theorem (Theorem~\ref{thm:main2})  that addresses the above-mentioned difficulties and bypasses both Banach and Lagrange-Good inversion. 
The novelty is two-fold. First, we work on the level of formal series and relate the formal inverse to generating functions of trees or equivalently, solutions of certain formal fixed point problems (Proposition~\ref{prop:trees}). This part is inspired by the combinatorial proof of the Lagrange-Good formula for finitely many variables given in~\cite{gessel1987}, we will consider this relation in more details in a forthcoming work. Second, we provide sufficient conditions for the convergence of the formal inverse, i.e., of a generalized tree generating functions (Theorem~\ref{thm:main1}). 
The inversion theorem is of an abstract general nature and has the potential of being applied to other situations than the density-activity relation in statistical mechanics. 

In our second group of results (Section~\ref{sec:virial}), we apply the abstract inversion theorem to the concrete problem of inverting the functional that maps the activity profile in an inhomogeneous grand-canonical Gibbs measure (or even a general multi-species system) to the density profile. We exhibit domains on which the activity profile is written as a convergent series in the density profile, relate the coefficients to two-connected graphs,  and show that the virial expansion for the pressure as a functional of the position-dependent density profile converges and is indeed given in terms of two-connected graphs (Theorem~\ref{thm:virmain2}). These results work for general stable pair potentials. 

Finally in Section~\ref{sec:examples} we apply the results to different more concrete choices of pair potentials. We demonstrate the power of our approach for systems of homogeneous hard spheres, our results yield a significant improvement over previously available bounds (Theorem~\ref{thm:hom-virmain}). For general non-negative potentials the improvement is almost 27\%. 
 For mixtures of thin rods with different orientiations, we obtain a series representation of the (grand-canonical) free energy as a function of the overall density $\rho_0$ of rods and the probability density $p(\sigma)$ on different orientations (Theorem~\ref{thm:onsager} and Corollary~\ref{cor:feinternal}). 
In fact, in an early work, Onsager \cite{onsager1949} derived a density functional for liquid crystals,
keeping track of the orientation of the atomistic elongated molecules. Working in the canonical ensemble
he discretized the space of orientations and assigned each value to a species obtaining
a multi-species canonical partition function for (for finitely many) species. Following the new developments \cite{pulvirenti-tsagkaro2012}, the convergence of this expansion can be easily proved to be valid in the low density regime. Our result allows for a direct treatment of continuous values of the orientation as inhomogeneous systems. It bypasses the need to estimate errors from discretizing the orientation space, at the price of a detour through the grand-canonical ensemble. The improvements we are obtain are purely due to the improved inversion results as we used the classical tree-graph bound in the grand-canonical ensemble \cite{penrose1963},  \cite{ruelle1969book},\cite{MalyshevMinlos1991}, \cite{poghosyan-ueltschi2009}, \cite{procacci-yuhjtman2017} and for marked systems \cite{kuna2001}.

%This, together with the fact that one would like to work closer to the common operations
%in liquid state theory which reveal the underlying combinatorial structure and give coefficients expressed via special graph classes
%(e.g. irreducible for the virial expansion)
%led us to seek for an alternative point-wise result.
%The key new element is a combinatorial identity between generating functions of trees. 
%Its validity can be also proved by a simple argument involving
%formal power series by showing that it comes as the unique solution of a finite system of equations. In the present paper we
%chose to give the latter explanation in order to emphasize the simplicity of the argument. However, for a better understanding, more details
%about the combinatorial structure and its connection to the Lagrange-Good inversion formulas will be given in a subsequent publication.

%Thus, having at hand the validity of density functionals for inhomogeneous systems, a number of interesting applications in the density function theory are rigorously tractable.

Following the above discussion
we summarize below the main outcomes of this paper:
\begin{enumerate}
\item Proof of a novel inversion theorem (Theorem~\ref{thm:main1}), applicable to the inversion of the density-activity relation for inhomogeneous systems (Section~\ref{sec:virial}), yielding a convergent power series of the inverse map.  
\item Key technical tool: a fixed point equation for generating functions of special trees (Proposition~\ref{prop:trees}).   
\item Various applications: inhomogeneous gas, liquid crystals,
	molecules with various shapes (internal degrees of freedom), see Section~\ref{sec:examples}.
\item Comparison to existing theorems of inversion in Banach spaces (Proposition~\ref{prop:rhoholo} and Theorem~\ref{thm:banach}). 
\item Discussion of  the improvement of the radius of convergence for the (homogeneous) hard sphere gas (Section~\ref{sec:hom}).
\end{enumerate}

\section{General inversion theorems}\label{sec:general}

\subsection{Main inversion theorem with proof}
Let $(\mathbb X, \mathcal X)$ be a measurable space and $\mathfrak M(\mathbb X,\mathcal X)$ the set of $\sigma$-finite non-negative measures on $(\mathbb X,\mathcal X)$. Further let $\mathfrak M_\mathbb C(\mathbb X, \mathcal X)$ be the set of complex linear combinations of measures in $\mathfrak M(\mathbb X,\mathcal X)$. When there is no risk of confusion, we shall write $\mathfrak M$ and $\mathfrak M_\mathbb C$ for short.  
Suppose we are given a family of measurable functions
$A_n: \mathbb X\times \mathbb X^n\to \C$, $(q, (x_1,\ldots, x_n))\mapsto A_n(q;x_1,\ldots,x_n)$. We assume that each $A_n$ is symmetric in the $x_j$'s, i.e., 
\be
	A_n(q;x_{\sigma(1)},\ldots,x_{\sigma(n)})= A_n(q;x_1,\ldots,x_n),
\ee
for all permutations $\sigma\in \mathfrak S_n$. 
 When we say that a power series converges absolutely, we mean that
 \be
	\sum_{n=1}^\infty \frac{1}{n!}\int_{\mathbb X^n} \bigl|A_n(q;x_1,\ldots,x_n)\bigr|\, |z|(\dd x_1)\cdots |z|(\dd x_n)< \infty
\ee
where $|z|$ is the total variation of $z$\footnote{If $z= \mu_1 - \mu_2 + \mathrm i \mu_3 - \mathrm i \mu_4$ with $\mu_1,\ldots, \mu_4$ mutually singular $\sigma$-finite non-negative measures, then $|z| = \sum_{i=1}^n \mu_i$.}
Let $\mathscr D(A)\subset \mathfrak M_\C$ be the domain of convergence of the associated power series, that is $z\in \mathscr D(A)$ if and only if the power series converges absolutely in the above sense.
We set 
\be \label{Adef}
	A(q;z):= \sum_{n=1}^\infty \frac{1}{n!}\int_{\mathbb X^n} A_n(q;x_1,\ldots,x_n) z(\dd x_1)\cdots z(\dd x_n) \qquad (z\in \mathscr D(A)). 
\ee 
We are interested in maps of the form
\be
	\mathfrak M_\C \supset \mathscr D(A) \to \mathfrak M_\C,\quad z \mapsto \rho[z]
\ee
 given by 
 \be \label{eq:rhoa}
 	\rho[z](\dd q) \equiv \rho(\dd q;z) := \e^{- A(q;z)} z(\dd q),
 \ee
where $\rho(\dd q;z)$ is just a notation for $\rho[z](\dd q)$. The latter is useful whenever one wants to stress the $q$ instead of the $z$ dependence. 
 Thus $\rho[z]$ is absolutely continuous with respect to $z$ with Radon-Nikod{\'y}m derivative $\exp(-  A(q;z))$. We want to determine the inverse map $\nu\mapsto \zeta [\nu]$, 
 $$ 
 	\nu =\rho[z] \, \Leftrightarrow\, z = \zeta[\nu]. 
$$
Suppose for a moment that such an inverse map exists. 
Clearly $z$  is equivalent to $\nu = \rho[z]$ with Radon-Nikod{\'y}m derivative $\exp(  A(q;z))$. Consequently we should have 
\be
	\zeta[\nu] (\dd q)  \equiv \zeta(\dd q; \nu) = \e^{ A(q; \zeta [\nu])} \nu(\dd q). 
\ee
This observation is the starting point for our inversion result, namely the family of power series $(T_q^\circ)_{q\in \mathbb X}$ given by 
\be
	T_q^\circ (\nu) \equiv T^\circ(q;\nu)=\e^{A(q;\zeta[\nu])} 
\ee 
should solve 
\be
	\zeta[\nu] (\dd q)  = T_q^\circ (\nu) \nu(\dd q)
= \e^{ A(q; \nu T_q^\circ (\nu))} \nu(\dd q)
\ee
and therefore
\be \label{tree-eq}\tag{$\mathsf{FP}$}
	T_q^\circ(\nu) = \exp\Biggl( \sum_{n=1}^\infty \frac{1}{n!} \int_{\mathbb X^n} A_n(q;x_1,\ldots, x_n) T_{x_1}^\circ(\nu)\cdots T_{x_n}^\circ(\nu) \nu(\dd x_1)\cdots \nu(\dd x_n)\Biggr). 
\ee 
In Proposition~\ref{prop:trees} below we provide a combinatorial interpretation of $T_q^\circ$ as the exponential generating function for colored rooted, labelled trees whose root is a ghost of color $q$ (i.e., the root does not come with powers of $\nu$ in the generating function). For our main inversion theorem, however, it is enough to know that the fixed point equation~\eqref{tree-eq} determines the power series $(T_q^\circ)_{q\in \mathbb X}$ uniquely.

\begin{lemma}\label{lem:treesolution}
	There exists a uniquely defined family of formal power series 
	\bes
		T_q^\circ(\nu) = 1+ \sum_{n=1}^\infty \frac{1}{n!} \int_{\mathbb X^n} t_n(q;x_1,\ldots, x_n) \nu(\dd x_1)\cdots \nu(\dd x_n)  \qquad (q\in \mathbb X)
	\ees
	with $t_n:\mathbb X\times\mathbb X^n\to \C$ measurable and symmetric in the $x_j$'s, 	that solves~\eqref{tree-eq} in the sense of formal power series. 
\end{lemma} 

As the above expressions are interpreted in the sense of formal power series, neither the series need to converge nor the integrals need to exist.
 
\begin{proof}
Set $t_0:=1$. 
	Let $B_n(q;x_1,\ldots,x_n)$ be the coefficients of the series in the exponential in~\eqref{tree-eq}, i.e., each $B_n:\mathbb X\times\mathbb X^n\to \C$ is measurable, and we have
	\begin{multline*}
		\sum_{n=1}^\infty \frac{1}{n!}\int_{\mathbb X^n} B_n(q;x_1,\ldots,x_n) \nu(\dd x_1)\cdots \nu(\dd x_n) \\
			= \sum_{n=1}^\infty \frac{1}{n!} \int_{\mathbb X^n} A_n(q;x_1,\ldots, x_n) T_{x_1}^\circ(\nu)\cdots T_{x_n}^\circ(\nu) \nu(\dd x_1)\cdots \nu(\dd x_n)
	\end{multline*}
	in the sense of formal power series. Then 
	\be \label{eq:Bn}
		B_n(q;x_1,\ldots,x_n) = \sum_{m=1}^n \sum_{\substack{J\subset [n]  \\ \#J =m}}  A_{m}\bigl(q;(x_j)_{j\in J}\bigr) \sum_{\substack{(V_j)_{j\in J}:\\  \dot \cup_{j\in J} V_j  = [n]\setminus J}} \prod_{j\in J} t_{\# V_j}\bigl(x_j; (x_v)_{v\in V_j}\bigr),
	\ee
	see Eq.~\eqref{fcomp2} in Appendix~\ref{ApFormal}.
	The third sum is over ordered partitions $(V_j)_{j\in J}$ of $[n]\setminus J$, indexed by $J$,  into $\#J$ disjoint sets $V_j$, with $V_j =\varnothing$ explicitly allowed. 
	For example, 
	\begin{align*}
		B_1(q;x_1) &= A_1(q;x_1),\\
		B_2(q;x_1,x_2) & = A_2(q;x_1,x_2) + A_1(q;x_1) t_1(x_1;x_2) + A_1(q;x_2) t_1(x_2;x_1).
	\end{align*}
	More generally, $B_n(q;\cdot)$ depends on $t_1(q;\cdot),\ldots,t_{n-1}(q;\cdot)$ alone. This is the only aspect of~\eqref{eq:Bn} that enters the proof of this lemma. 
	
	For $n\in \N$, let $\mathcal P_n$ be the collection of set partitions of $\{1,\ldots,n\}$. 	The family $(T_q^\circ)_{q\in \mathbb X}$ solves~\eqref{tree-eq} in the sense of formal power series if and only if for all $n\in \N$ and $q,x_1,\ldots, x_n\in \mathbb X^n$, we have 
 	\be \label{eq:tnBn}
 		t_n(q;x_1,\ldots,x_n) =\sum_{m=1}^n \sum_{\{J_1,\ldots, J_m\}\in \mathcal P_n} \prod_{\ell=1}^m B_{\#J_\ell}\bigl( q;(x_j)_{j\in J_\ell}\bigr),
 	\ee
 	see Eq.~\eqref{eq:formal-exponential} in Appendix~\ref{ApFormal}. In particular, 
 	\begin{align*}
 		t_1(q;x_1) & = B_1(q;x_1) = A_1(q;x)\\
 		t_2(q;x_1,x_2) & = B_2(q;x_1,x_2) + B_1(q;x_1) B_1(q;x_2) 
	\end{align*} 
	which determines $t_1$ and $t_2$ uniquely. A straightforward induction over $n$, exploiting that  the right-hand side of~\eqref{eq:tnBn} depends on $t_1,\ldots, t_{n-1}$ alone (via $B_1$,\ldots, $B_n$), shows that the system of equations~\eqref{eq:tnBn} has a unique solution $(t_n)_{n\in \N}$. 
\end{proof}
 
\begin{remark}
	The proof of Lemma~\ref{lem:treesolution} shows that the coefficients $(t_n)_{n\in \N}$ can be computed recursively.  
\end{remark}

\noindent Next we provide a sufficient condition for the absolute convergence of the series $T_q^\circ(\nu)$. 

\begin{theorem}\label{thm:main1}
	Let $T_q^\circ(\nu)$ be the unique solution of~\eqref{tree-eq} from Lemma~\ref{lem:treesolution}. Assume 
	 that for some measurable function $b:\mathbb X\to [0,\infty)$, the measure $\nu\in \mathfrak M_\C$ satisfies, for all $q\in \mathbb X$, 
	\be \label{suff1} \tag{$\mathscr S_b$}
		\sum_{n=1}^\infty \frac{1}{n!}\int_{\mathbb X ^n }|A_n(q;x_1,\ldots,x_n)|\e^{\sum_{j=1}^n b(x_j)} |\nu|(\dd x_1)\cdots |\nu|(\dd x_n) \leq b(q).
	\ee
 Then, for all $q \in \mathbb X$, we have that
	\begin{equation}\label{mainest1} \tag{$\mathscr M_b$}
1 + \sum_{n=1}^{\infty}\frac{1}{n!}\int_{\mathbb X^n}  
\left| t_{n}(q;x_1,\ldots, x_n)\right| \ |\nu|(\dd x_1) \cdots |\nu|(\dd x_n) \leq \e^{b(q)}
\end{equation}
	and the fixed point equation~\eqref{tree-eq} holds true as an equality of absolutely convergent series.
\end{theorem}

\begin{proof}
	The inductive proof is similar to~\cite{ueltschi2004,poghosyan-ueltschi2009}.	Let $S_q^N(\nu)$, $N\in \N_0$, be the partial sums for the left-hand side of~\eqref{mainest1},
\bes
	S_q^N(\nu) :=1+ \sum_{n=1}^N \frac{1}{n!} \int_{\mathbb X^n}\left| t_{n}(q;x_1,\ldots, x_n)\right| \   |\nu| (\dd x_1)\cdots  |\nu|(\dd x_n).
\ees
	We prove $S_q^N(\nu)\leq \e^{b(q)}$ by induction on $N$, building on the proof of Lemma~\ref{lem:treesolution}. The estimate for the full series then follows by a passage to the limit $N\to \infty$.
	
	For $N=0$, we have $S_q^0(\nu) =1$ and the inequality $S_q^0(\nu) \leq \exp(b(q))$ is trivial. Now assume $S_q^{N-1}(\nu)\leq \exp(b(q))$. The triangle inequality applied to Eqs.~\eqref{eq:Bn} and~\eqref{eq:tnBn} yields the same iterative formula for $\left| t_{n}(q;x_1,\ldots, x_n)\right|$ as for $ t_{n}(q;x_1,\ldots, x_n)$ just with $ A_n(q;x_1,\ldots,x_n)$ replaced by $\bigl| A_n(q;x_1,\ldots,x_n) 	\bigr|$.  We noted before that, if we consider $S_q^N(\nu)$ and hence only $\left| t_{n}(q;x_1,\ldots, x_n)\right|$ for $n \leq N$, then on the right hand side only $\left| t_{n}(q;x_1,\ldots, x_n)\right|$  with $n \leq N-1$ appear. However, there are some terms on the right hand side, which as well only contain $\left| t_{n}(q;x_1,\ldots, x_n)\right|$  with $n \leq N-1$  but which come from some term $\left| t_{n}(q;x_1,\ldots, x_n)\right|$  on the left hand side for $n >N$. Adding these missing terms, we reconstruct an exponential on the right hand side. As all of these additional terms are non-negative, we get the following inequality, instead of an equality
	\begin{align*}
		S_q^N(\nu) & \leq \exp\Biggl( \sum_{n=1}^{N-1} \frac{1}{n!} \int_{\mathbb X^n}\bigl| A_n(q;x_1,\ldots,x_n) 	\bigr|\, S_{x_1}^{N-1}(\nu)\cdots S_{x_n}^{N-1}(\nu)
		\   |\nu| (\dd x_1)\cdots  |\nu|(\dd x_n)
		\Biggr) \\
		& \leq  \exp\Biggl( \sum_{n=1}^{N-1} \frac{1}{n!} \int_{\mathbb X^n}\bigl| A_n(q;x_1,\ldots,x_n)\bigr| 	\e^{b(x_1)+\cdots+ b(x_n)}
		\   |\nu| (\dd x_1)\cdots  |\nu|(\dd x_n)
		\Biggr)\\
		&\leq \e^{b(q)}.
	\end{align*}
	The induction is complete. It follows that~\eqref{mainest1} holds true. In particular, the series $T_q^\circ(\nu)$ is absolutely convergent and satisfies $|T_q^\circ(\nu)|\leq \exp(b(q))$. By condition~\eqref{suff1}, the right-hand side of the fixed point equation~\eqref{tree-eq} is absolutely convergent as well. Therefore Eq.~\eqref{tree-eq} holds true not only as an identity of formal power series but in fact as an identity of well-defined complex-valued functions. 
\end{proof} 

\begin{remark}
	For non-negative functions $A_n$, the convergence estimate is sharp, in the following sense: If $\nu\in \mathfrak M$ is a non-negative measure and $T_q^\circ(\nu)<\infty$, 
	then there exists a function $b:\mathbb X\to [0,\infty)$ such that~\eqref{mainest1} holds true. Indeed, an induction over $n$, based on Eqs.~\eqref{eq:Bn} and~\eqref{eq:tnBn}, shows that if the $A_n$'s are non-negative, then the coefficients $B_n$ and $t_n$ are non-negative as well. If $T_q^\circ(\nu)<\infty$, we may define 
	$$
		b(q):=\log T_q^\circ(\nu).
	$$
	Notice $b(q)\geq 0$ because of $T_q^\circ(\nu)\geq 1$ for non-negative $t_n$ and $\nu$. It follows from~\eqref{tree-eq} that the inequality~\eqref{suff1} holds true and is in fact an equality.  This was already noticed in  \cite[Proposition~2.9]{jansen2018pointproc} and the proof of Theorem~4.2(b) in~\cite{jansen2015tonks}. 
\end{remark}

Now that we have addressed the convergence of the series $T_q^\circ$, we may come back to the inversion of the map $\mathscr D(A)\ni z\mapsto \rho[z]$. For measurable $b:\mathbb X\to [0,\infty)$, let 
\be \label{eq:Db}
	\mathscr V_b:=\{ \nu \in \mathfrak M_\C\mid \nu\ \text{satisfies condition~\eqref{suff1}}\}.
\ee
For $\nu \in \mathscr V_b$, define $\zeta[\nu]\in \mathfrak M_\C$ by 
\be \label{zetadef}
	\zeta[\nu](\dd q) = \zeta(\dd q;\nu):= T_q^\circ(\nu) \nu(\dd q). 
\ee

\begin{theorem}\label{thm:main2}
	For every weight function $b:\mathbb X\to \R_+$, there is a set $\mathscr U_b\subset \mathscr D(A)$ such that $\rho:\mathscr U_b\to \mathscr V_b$ is a bijection with inverse $\zeta$. 
\end{theorem}	

\begin{proof}
	Let $\mathscr U_b$ be the image of $\mathscr V_b$ under $\zeta$. By Theorem~\ref{thm:main1}, the set $\mathscr U_b$ is contained in $\mathscr D(A)$, in particular if $z=\zeta[\nu]$ with $\nu \in \mathscr V_b$, then $\rho[z]$ is well-defined with 
	\begin{align*}
		\rho(\dd q;z) & = \e^{-A(q;z)} z(\dd q) = \e^{-A(q;\zeta[\nu])} \zeta(\dd q;\nu)  \\
			& =  \e^{- A(q;\zeta[\nu])} T_q^\circ(\nu)  \nu(\dd q)= \nu(\dd q). 
	\end{align*}
	For the last identity we have used the fixed point equation~\eqref{tree-eq}.
Thus we have checked that if $z = \zeta[\nu]$, with $\nu\in \mathscr V_b$, then $\rho[z]= \nu$. 
	  Conversely, if $\nu = \rho[z]$ with $z\in \mathscr U_b$, then by definition of $\mathscr U_b$ there exists $\mu \in \mathscr V_b$ such that $z = \zeta[\mu]$, hence $\nu = \rho[z] = \rho [\zeta[\mu]] = \mu\in \mathscr V_b$  and $z = \zeta[\mu]= \zeta[\nu]$. 	
\end{proof} 

\noindent Finally we provide a combinatorial formula for the function $T_q^\circ(\nu)$ appearing in the inverse $\zeta[\nu]$. Consider a genealogical tree that keeps track not only of mother-child relations, but also of groups of  siblings born at the same time. This results in a tree for which children of a vertex are partitioned into cliques (singletons, twins, triplets, etc.). Accordingly for $n\in \N$ we define $\mathcal{TP}_n^\circ$ as the set of pairs $(T, (P_i)_{0\leq i \leq n})$ consisting of: 
\begin{itemize}
	\item A tree $T$ with vertex set $[n]:=\{0,1,\ldots,n\}$. The tree is considered rooted in $0$ (the ancestor).
	\item For each vertex $i\in \{0,1,\ldots,n\}$, a set  partition $ P_i$ of the set of children\footnote{The members of the partition are assumed to be non-empty, except we consider the partition of the empty set.} of $i$. If $i$ is a leaf (has no children), then we set $P_i = \varnothing$. 
\end{itemize} 
For $x_0,\ldots, x_n\in \mathbb X$, we define the weight of an enriched tree $(T, (P_i)_{0\leq i \leq n})\in \mathcal{TP}_n^\circ$ as 
\be \label{eq:treeweight}
	w\bigl( T,(P_i)_{0\leq i \leq n}; x_0,x_1,\ldots, x_n\bigr)  := \prod_{i=0}^n \prod_{J\in P_i} A_{\#J +1}\bigl(x_i;(x_j)_{j\in J}\bigr)
\ee
with $\prod_{J\in \varnothing} =1$. So the weight of an enriched tree is a product over all cliques of twins, triplets, etc., contributing each a weight that depends on the variables $x_j$ of the clique members and the variable $x_i$ of the parent. 

\begin{prop} \label{prop:trees}
	The family of power series $(T_q^\circ)_{q\in\mathbb X}$ from Lemma~\ref{lem:treesolution} is given by 
	$$
		T_q^\circ(z) = 1+ \sum_{n=1}^\infty \frac{1}{n!} \int_{\mathbb X^n} \sum_{(T,(P_i)_{i=0,\ldots,n})\in \mathcal{TP}_n^\circ} 
		w\bigl( T,(P_i)_{i=0,\ldots,n}; q,x_1,\ldots, x_n\bigr) z^n(\dd \vect x). 
	$$
\end{prop} 

\begin{proof}
	We check that the generating function of the weighted enriched trees satisfies~\eqref{tree-eq}. Functional equations for generating functions of labelled trees are standard knowledge~\cite{bergeron-labelle-leroux1998book}, we provide a self-contained proof for the reader's convenience. 
	Define 
	$$
		\tilde t_n(q;x_1,\ldots,x_n):=  \sum_{(T,(\mathcal P_i)_{i=0,\ldots,n})\in \mathcal{TP}_n^\circ} 
		w\bigl( T,(P_i)_{0\leq i \leq n}; q,x_1,\ldots, x_n\bigr).
	$$
	Further define $\tilde B_n(q;x_1,\ldots,x_n)$ but restricting the sum to enriched trees for which $\#P_0 =1$ (all children of the root belong to the same clique). Further set $t_0 =1$ and $\tilde B_0 =0$. 
	For $V\subset \N$ a finite non-empty set, define $\mathcal{TP}^\circ(V)$ in the same way as $\mathcal {TP}^\circ_n$ but with $\{0,1,\ldots,n\}$ replaced by $\{0\} \cup V$.	For $V=\varnothing$ we define $\mathcal{TP}^\circ(V) =\varnothing$ and assign the empty tree the weight $1$. For non-empty trees, weights $w(R;(x_j)_{j\in V\cup\{0\}})$ are defined in complete analogy with~\eqref{eq:treeweight}. 
	
		Clearly there is a bijection between enriched trees $R\in \mathcal{TP}_n^\circ$ and set partitions $\{J_1,\ldots,J_m\}$ of $[n]:=\{ 1, \ldots , n\}$ together with enriched trees $R_i\in \mathcal {TP}^\circ (J_i)$, $i=1,\ldots,m$ for which all children of the root are in the same clique. Indeed, the number $m$ corresponds to the number of cliques in which the children of the root are divided and the blocks $J_1,\ldots, J_m$ group descendants of the root, where $J_k$ contains the children of the root which are in the $k$-th. clique and all their decedents. The weight of an enriched tree $R$ is equal to the product of the weights of the subtrees $R_i$. Therefore 
	\be \label{eq:tnBntilde}
 		\tilde t_n(q;x_1,\ldots,x_n) =\sum_{m=1}^n \sum_{\{J_1,\ldots, J_m\}\in \mathcal P_n} \prod_{\ell=1}^m \tilde B_{\#J_\ell}\bigl( q;(x_j)_{j\in J_\ell}\bigr). 
 	\ee
Furthermore there is a one-to-one correspondence between, on the one hand, enriched trees where all the children of the root are in the same clique and on the other hand tuples $(J, (V_j)_{j\in J}, (R_j)_{j\in J})$ consisting of non-empty set $J\subset [n]$, an ordered partition $(V_j)_{j\in J}$ of $[n]\setminus J$ (with $V_j =\varnothing$ allowed), and a collection of enriched trees $R_j\in \mathcal {TP}^\circ(V_j)$. Overall, $J$ and $(V_j)_{j\in J}$ give a partition of $[n]$. The set $J$ consists of the labels of the children of the root, that is the one clique which all these children form and for each $j \in J$, the set $V_j$ consists of the labels of the descendants of $j$. ($V_j= \varnothing$ means that $j$ is a leave of the tree) It follows that 
	\be \label{eq:Bntilde}
		\tilde B_n(q;x_1,\ldots,x_n) = \sum_{m=1}^n \sum_{\substack{J\subset [n]  \\ \#J =m}}  A_{m}\bigl(q;(x_j)_{j\in J}\bigr) \sum_{\substack{(V_j)_{j\in J}:\\  \dot \cup_{j\in J} V_j  = [n]\setminus J}} \prod_{j\in J} \tilde t_{\# V_j}\bigl(x_j; (x_v)_{v\in V_j}\bigr). 
	\ee
	It follows from Eqs.~\eqref{eq:tnBntilde} and~\eqref{eq:Bntilde} that the formal power series with coefficients $\tilde t_n$ solves~\eqref{tree-eq}, therefore  Lemma~\ref{lem:treesolution} yields $\tilde t_n = t_n$. 
\end{proof}

\subsection{Scale of Banach spaces. Banach inversion}  \label{sec:banach} 
Formally, one is tempted to say that $\rho[z]$ is given by a power series with leading order $z$, hence differentiable with derivative at the origin given by the identity matrix; therefore the existence and regularity of the inverse map should follow from some general inverse function theorem. When $\mathbb X$ is finite so that $z$ can be identified with a finite vector $(z_x)_{x\in \mathbb X}\in \C^n$, with $n=\#\mathbb X$, this can be implemented and is indeed a standard ingredient for the virial expansion for single-species systems~\cite{lebowitz-penrose1964virial}. 

For infinite spaces $\mathbb X$ one may try a Banach inversion theorem. This works in some cases (see Theorem~\ref{thm:banach} below), but there are situations where the Banach inversion theorem is doomed to fail,
as illustrated by the following example. The example is inspired by concrete features of the multi-species Tonks model~\cite{jansen2015tonks} for rods of unbounded lengths $\ell_k = k$. 

\begin{example} \label{ex:unbounded} 
	Let $\mathbb X= \N$ and identify measures on $\mathbb X$ with sequences $(z_k)_{k\in \N}$. Consider the map $(z_k)\mapsto (\rho_k)$ given by 
	$$
		\rho_1 = z_1,\quad \forall k \geq 2: \ \rho_k = z_k \exp(- k z_1). 
	$$
	Let $\ell^\infty(\N)$ be the space of bounded complex-valued sequences equipped with the supremum norm and  $X_{c}$ the space of sequences $(\nu_k)$ with $||\nu||_c:= \sup_{k\in \N}|\nu_k| \exp( - ck) <\infty$, for some fixed scalar $c>0$. 
	We may view $(z_k)\mapsto (\rho_k)$ as a map from the open ball $B(0,c)\subset \ell^\infty(\N)$ to $X_c$. 
	The derivative $\mathrm D\rho(0)$ is the identity map or more precisely, the embedding $\iota:\ell^\infty(\N)\to X_c$, $\iota(h):= h$. It is injective and continuous but it does not have a continuous inverse, therefore Banach inversion theorems are not applicable. The issue arises because the norms $||\cdot||_\infty$ and $||\cdot||_c$ are not equivalent. A target space with inequivalent norm is needed because, for every $z_1<0$---no matter how small---$|\rho_k|\gg |z_k|$ as $k\to \infty$.
\end{example}	
	
It turns out that the natural analytic framework for our inversion theorem uses not a single Banach space, but instead a scale of Banach spaces, as is the case for the Nash-Moser theorem~\cite{hamilton1982, secchi2016}. We explain this aspect in more detail here as this clarifies the issues raised in~\cite[Section 2.2]{jttu2014} and~\cite[Theorem 2.8]{jansen2015tonks}.

Let us fix a reference measure $m\in \mathfrak M(\mathbb X, \mathcal X)$  and we restrict to measures that are absolutely continuous with respect to $m$. Remember that  $\rho[z](\dd x)$ is absolutely continuous  with respect to the measure $z(\dd x)$, so if $z$ is absolutely continuous with respect to $m$, then so is $\rho[z]$. We work with the Radon-Nikod{\'y}m derivatives rather than the measures and write 
$$
	z(\dd x) = z(x) m(\dd x),\quad \rho(\dd x;z) = \rho(x;z) m(\dd x), 
$$
similarly for $\nu$ and $\zeta$. Fix a weight function $b:\mathbb X\to \R_+$ and assume that $m$ satisfies condition~\eqref{suff1}. Let $L^\infty(\mathbb X,m)$ be the space of bounded functions (precisely, equivalence classes up to $m$-null sets), equipped with the supremum norm 
$$
		||h||_\infty:= \esssup_{x\in \mathbb X} |h(x)|.
$$
Write $B_r(0)$ for open balls of radius $r$ centered at $0$.
For $h:\mathbb X\to \C$ measurable and $k\in \Z$, define the weighted supremum norm
$$
		||h||_{kb}:= ||\e^{kb}h ||_\infty = \esssup_{x\in \mathbb X} |h(x)| \e^{kb(x)}
$$
and let $Y_{kb}$ be the associated Banach space. Notice the inclusions 
$$
	\ldots \subset Y_{2b}\subset Y_b\subset L^\infty(\mathbb X,m) \subset Y_{-b}\subset Y_{-2b}\subset \ldots 
$$
When $b$ is essentially bounded, then the inclusions are equalities and the norms $||\cdot ||_{kb}$, $||\cdot||_\infty$ are equivalent. For $||b||_\infty = \infty$, the inclusions are strict and the norms are inequivalent. 
Let $B(0,r)$ and $B_{kb}(0,r)$ be the open balls of radius $r$, centered at the origin, in  $L^\infty(\mathbb X,m)$ and $Y_{kb}$, respectively. 

\begin{prop}  \label{prop:rhoholo}
	Assume that $m\in \mathfrak M$ satisfies condition~\eqref{suff1}. Then the maps 
	\begin{align*}
		\rho &:\  B_{kb}(0,1) \to B_{(k-1)b}(0,1)\qquad (k\geq - 1)\\
		\zeta &:\  B_{kb}(0,1) \to B_{(k-1)b}(0,1)\qquad (k\geq 0)
	\end{align*} 
	are holomorphic, as maps between the Banach spaces $Y_{kb}$ and $Y_{(k-1)b}$. Moreover we have $\rho[\zeta[\nu]] = \nu$ and $\zeta[\rho[z]] = z$ for all $\nu\in B(0,1)$ and $z \in B_{-b}(0,1)$. 
\end{prop} 

\noindent The proposition is proven at the end of this section. The inclusions $\rho[B_{kb}(0,1)]\subset B_{(k-1)b}(0,1)$ and 
$\zeta[B_{kb}(0,1)]\subset B_{(k-1)b}(0,1)$ follow from the inequalities
\be \label{rzbbound}
	|\rho(q;z)|\leq |z(q)|\e^{b(q)}, \qquad |\zeta(q;\nu)|\leq |\nu(q)|\e^{b(q)},
\ee
valid for all $z \in \overline{B_{-b}(0,1)}$, $\nu \in \overline{B(0,1)}$, and all $q\in \mathbb X$, assuming $m$ satisfies~\eqref{suff1} by using Theorem~\ref{thm:main1}. The difference to the previous results is that we show here  \emph{uniform} convergence of  the power series expansions of $\rho$ and $\zeta$ in the relevant norms.  

We briefly check~\eqref{rzbbound}. 
If  $z\in \overline{B_{-b}(0,1)}$ then $|z(q)|\leq ||z\e^{-b}||_\infty \e^{b(q)}\leq \e^{b(q)}$ for $m$-almost all $q$. Since $m$ satisfies condition~\eqref{suff1}, it follows by Theorem~\ref{thm:main1} that the measure $z(\dd q) = z(q) m(\dd q)$  is in the domain of convergence $\mathscr D(A)$ of $A$ (though it does not fulfill condition~\eqref{suff1}) and $|A(q;z)| \leq b(q)$, consequently 
$|\rho(q;z)|\leq \e^{b(q)} |z(q)|$. If $\nu \in \overline{B(0,1)}$, then, using again that $m$ satisfies condition~\eqref{suff1}, we see that in this case the measure $\nu(\dd q) = \nu(q) m(\dd q)$ satisfies condition~\eqref{suff1} as well and the bound~\eqref{mainest1} yields $|\zeta(q;\nu)| = |\nu(q) T_q^\circ(\nu)|\leq |\nu(q)|\e^{b(q)}$. 

It is an immediate consequence of Proposition~\ref{prop:rhoholo} that $\rho$ is a bijection from $\mathcal U_b:= \zeta[B(0,1)]\subset B_{-b}(0,1)$ onto $B(0,1)$. If $b$ is essentially bounded, then all norms are equivalent, hence $\rho$ and $\zeta$ are holomorphic as maps in $L^\infty(\mathbb X,m)$ and $\mathcal U_b=\rho^{-1}(B(0,1))$ is open in the non-weighted sup norm $||\cdot||_\infty$. Moreover we have the inclusion
$$
	\mathcal U_b\subset\{z:\, ||z\e^{-b}||_\infty< 1\}\subset \{z:\, ||z||_\infty < \e^{||b||_\infty}\}
$$
and we obtain the following corollary. 

\begin{cor} \label{cor:banach}
	Assume that $m\in \mathfrak M$ satisfies condition~\eqref{suff1} and in addition $||b||_\infty <\infty$. Then 
	$\rho[\cdot]$ maps some open subset $\mathcal U_b$  of $B(0,\e^{||b||_\infty})\subset L^\infty(\mathbb X,m)$ biholomorphically onto $B(0,1)$, and the inverse map is $\zeta$. 
\end{cor} 

Corollary~\ref{cor:banach} points out a situation where Banach inversion does work, which raises the question whether a similar result can be obtained directly, bypassing the introduction of a weight function $b$. This is indeed possible. 
Let us fix a reference measure $m$ as before but drop the requirement that $m$ satisfies \eqref{suff1}. 
Set
\be \label{suffbounded}
	M(r):= \esssup_{q\in \mathbb X}\ \sum_{n=1}^\infty \frac{r^n}{n!} \int_{\mathbb X^n} \bigl|A_n(q;x_1,\ldots, x_n)\bigr| \,   m(\dd x_1)\cdots m(\dd x_n)<\infty.
\ee
and let 
\be \label{Rdef}
	R:=\sup\{r\geq 0\mid M(r)<\infty\}.
\ee

\begin{theorem} [Banach inversion] \label{thm:banach}
	Assume that~\eqref{suffbounded} holds true for some $r>0$ and let 
	$R>0$ be as in~\eqref{Rdef}. Let 
	$$
		P:=\frac18\,  \sup_{0<r< R}  r \e^{-M(r)}.
	$$
	Then the functional $\rho$ maps some open neighborhood of the origin $\mathcal O\subset B(0,R)\subset L^\infty(\mathbb X,m)$ biholomorphically onto the open ball $B(0,P)$. 	
\end{theorem} 

\begin{proof}
	The map $\rho: B(0,R)\to L^\infty(\mathbb X,m)$ is holomorphic. The proof of the holomorphicity is similar to the proof of Proposition~\ref{prop:rhoholo} and therefore omitted. 
	The derivative at the origin is the identity: $\mathrm D\rho(0) = \mathrm{id}$. On $B(0,r)\subset B(0,R)$, the map is bounded by $r\exp( M(r))$. 	Therefore, by Theorem~\ref{thm:holoinv}, for each $r\in (0,R)$, the functional $\rho$ maps 	 the open ball $B(0,\frac14 r\e^{-M(r)})\subset L^\infty(\mathbb X,m)$ biholomorphically onto a domain covering $B(0,\frac18 r \e^{- M(r)})$. We optimize over $r$ and obtain the theorem.
\end{proof} 

\begin{remark}
Let us compare the radius of convergence of the inverse function $P$ which we obtained with the technique of Banach inversion theorem with the convergence results we obtain with the new inversion technique described in Section~\ref{sec:general}, namely Corollary~\ref{cor:banach}. Let us call $P'$ the radius of convergence in the latter case. We will show that $P' =8P$ and thus even in those situations where a direct application of Theorem~\ref{thm:holoinv} is possible, it yields a bound that is worse than ours.
 
Let us first derive an expression of $P'$ in terms of $M$ as defined in \eqref{suffbounded}. If $m$ satisfies condition~\eqref{suff1} with $||b||_\infty <\infty$, then $M(1) \leq ||b||_\infty<\infty$. Conversely, assume $M(s)<\infty$ for some $s>0$ and consider constant weight functions $b(q)\equiv b>0$. Then, for every $b>0$, choosing $s>0$ small enough we may assume $M(s\e^b)\leq b$ and then the rescaled measure $sm$ satisfies condition~\eqref{suff1}. Noting that 
$$
	\left\{\mu \in \mathfrak M_\C:\, \left\|\frac{\dd \mu}{\dd (sm)}\right\|_\infty <1\right\} = \left\{\mu \in \mathfrak M_\C:\, \left\|\frac{\dd \mu}{\dd m}\right\|_\infty <s\right\},
$$
we deduce from Corollary~\ref{cor:banach} that $B(0,s)$ is contained in the domain of convergence of the density expansions. An optimization over $b$ and $s$ shows that the domain of convergence contains the open ball $B(0,P')$ with radius 
$$
	P':= \sup_{b>0}\sup\{s>0 \mid M(s\e^b) \leq b\}. 
$$
Below we check that $P'=8 P$. 
\end{remark} 

\begin{proof}[Proof of $P' = 8 P$]
	Let $\eps>0$ and $s\geq P'- \eps$. By definition of $P'$, there exists $b>0$ such that $M(s\e^b) \leq b$. Set $r:= s\e^b$. Then $M(r)\leq b<\infty$, thus $r\leq R$ and 
	$$
		r\e^{-M(r)} \geq r \e^{-b} = s\geq P'-\eps.
	$$
	It follows that $8 P\geq P'$. Conversely, let $s \geq 8P-\eps$. By definition of $P$ there exists $r\in (0,R)$ such that $s\leq r\exp( - M(r))$, hence $1\leq \exp( M(r))\leq \frac{r}{s}$. Set $b:= \log \frac{r}{s}$, then $b\geq 0$, $r= s\e^b$, and 
	$$M(s\e^b) = M(r) \leq \log \frac{r}{s}\leq b. $$
	It follows that $P'\geq s \geq 8 P-\eps$. We let $\eps\searrow0$ and deduce $P' = 8P$. 
\end{proof} 

\begin{proof}[Proof of Proposition~\ref{prop:rhoholo}]
	We only need to prove that the maps are holomorphic. 
	Consider first the map $\rho$.
	We have $\rho(q;z) = z(q)\mathcal E(q;z)$ with 
	\be \label{eexp}
		\mathcal E[q](z) = \mathcal E(q;z)=1+\sum_{n=1}^\infty \frac{1}{n!}\int_{\mathbb X^n} E_n(q;x_1,\ldots,x_n) z(x_1)\cdots z(x_n) m^n(\vect \dd x)
	\ee
	and 
	$$
		E_n(q;x_1,\ldots,x_n) = \sum_{m=1}^n\sum_{\{V_1,\ldots, V_m\}\in \mathcal P_n} \prod_{\ell=1}^m A_{\#V_\ell}\bigl( (x_j)_{j\in V_\ell}\bigr),
	$$
	see Appendix~\ref{ApFormal}, Eq.~\eqref{eq:formal-exponential}. We show first that $\mathcal E:B_{-b}(0,1)\to Y_{-b}$ is holomorphic, by proving that the series~\eqref{eexp} converges uniformly in the relevant operator norms. 
	Set 
	$$
		M_0(q;r):=1+\sum_{n=1}^\infty \frac{r^n}{n!}\int_{\mathbb X^n} \bigl| E_n(q;x_1,\ldots,x_n) \bigr| \e^{b(x_1)+\cdots + b(x_n)}  m^n(\vect \dd x).
	$$
	Then for all $r\in [0,1]$, we have 
	\be \label{eq:m0bound}
		M_0(q;r) \leq 	\exp \Biggl( \sum_{n=1}^\infty \frac{r^{n}}{n!} \int_{\mathbb X^n}\bigl| A_{n}(q;x_1 , \ldots , x_{n})\bigr|\, \e^{b(x_1)+\cdots b(x_{n})} m^{n}(\dd \vect x)\Biggr)\leq \e^{b(q)}
	\ee
	because $m$ satisfies condition~\eqref{suff1}. In particular, the power series $r\mapsto M_0(q;r)$ has radius of convergence $R\geq 1$. It follows from Cauchy's inequality for the Taylor coefficients of the series that for all $n\in \N$,
	$$
	 \frac{1}{n!}\Bigg|\frac{\partial^n M_0}{\partial r^n} (q;0)\Biggr| \leq  \sup_{t\in \C: |t|=1}|M_0(q;t)| = M_0(q;1)\leq \e^{b(q)}. 
	$$
	Therefore, we can bound
	$$
		\frac{1}{n!}\int_{\mathbb X^n} \bigl|E_n(q;x_1,\ldots,x_n) z(x_1)\cdots z(x_n)\bigr| m^n(\dd \vect x) \leq    \frac{ ||z\e^{-b}||_\infty^n}{n!}\Bigg|\frac{\partial^n M_0}{\partial r^n} (q;0)\Biggr| 
			\leq \e^{b(q)} ||z\e^{-b}||_\infty^n . 
	$$
	As a consequence, the map  $P_n$ defined on $Y_{-b}$  given by 
	$$
		P_n[z](q):= \frac{1}{n!} \int_{\mathbb X^n} E_n(q;x_1,\ldots, x_n) z(x_1)\cdots z(x_n) m^n(\dd \vect x) 
	$$
	satisfies for any $s \in \mathbb{R}$ (we will choose $s$ appropriately at the end)
	\be \label{polybound}
		||\e^{sb} P_n[z]||_{\infty} \leq || e^{(s+1)b}||_\infty \ ||\e^{-b}z ||_\infty^n. 
	\ee
	It follows from the polarization formulas, see e.g. ~\cite{mujica2006}, that  the associated multilinear map from $Y_{-b}^n$ to $Y_{sb}$ 
	is bounded, whenever $s \leq -1$ or $|| b ||_\infty < \infty$, hence $P_n$ is a continuous $n$-homogeneous polynomial (see Definition~\ref{def:holo1}). 
	By~\eqref{polybound}, the series $\mathcal E[z]= \sum_{n=1}^\infty P_n[z]$ converges uniformly in $||\e^{-b} z||_\infty\leq 1$. Therefore, the map 	$z \mapsto \mathcal E[z]$ as a map
	\be \label{ebmap}
		Y_{-b}\supset\{z:\, ||z\e^{-b}||_\infty< 1\}\to Y_{sb},
	\ee
	is holomorphic. For $k\geq -1$, it is also holomorphic as a map
	\be \label{ebmap2}
				Y_{k b}\supset\{z:\, ||z\e^{k b}||_\infty< 1\}\to Y_{sb},
	\ee
	because $Y_{kb}\subset Y_{-b}$ and $||z\e^{-b}||_\infty\leq ||z\e^{kb}||_\infty$. 
	
	Now we return to $\rho(q;z) =z(q) \mathcal E(q;z)$. By~\eqref{eq:m0bound}, we have 
	$$
		|\rho(q;z)| \leq |z(q)|\, |\mathcal E(q;z)|\leq  |z(q)| M_0(q,||ze^{-b}||_\infty)  \leq |z(q)\e^{b(q)}|
	$$
	hence 
	\be \label{oboe}
		||\e^{sb}\rho(z)||_\infty \leq ||z\e^{(s+1)b}||_\infty \leq ||z\e^{kb}||_\infty \leq 1 
	\ee
	whenever $s+1 \leq k$ and  $||z\e^{kb}||_\infty \leq1$. In order to prove the differentiability, let us introduce 
	$$ 
		\bigl(L_z h\bigr):= h(q) \mathcal E(q;z) + z(q) \bigl(\mathrm D\mathcal E(z) h\bigr)(q),
	$$
	which will be shown to be the derivative of $\rho[z]$. Using Cauchy's inequality  we get from the holomorphicity of $\mathcal  E(z)$ that there exists a $C >0$ with 
	  $||\e^{rb} \mathrm D\mathcal E(z) h||_\infty \leq C||h\e^{-b}||_\infty \leq C||h \e^{kb}||_\infty$, whenever $r \leq -1$. Then for $s-k \leq -1$   we get that
	\begin{align*}
		& |\e^{sb} L_z h||_\infty  \leq \left\| e^{kb}|h| \, e^{(s-k)b}| \mathcal E(z)|\, \right\|_\infty
 + \left\|  |z|e^{kb} \, e^{(s-k)b}\bigl|\mathrm D\mathcal E(z) h\bigr| \right\|_\infty 
  \leq \left\| e^{kb} h  \right\|_\infty
 + \left\|   z e^{kb} \right\|_\infty C||h \e^{kb}||_\infty  .\\ & 
\end{align*}
	Thus $L_z: Y_{kb}\to Y_{sb}$ is bounded. Let us show differentiability directly. Write 
	$$ \rho(z+ h) -\rho(z)  =  L_z h  + h \left( \mathcal E(z+h) - \mathcal E(z) \right) + z \left( \mathcal E(z+h)  -  \mathcal E(z)  - \bigl(\mathrm D\mathcal E(z) h\bigr) \right),
	$$
	which can be estimated as 
	\begin{align*}
		& ||\e^{sb} \bigl( \rho(z+h) - \rho(z) - L_z h\bigr) ||_\infty  \leq  ||h \e^{kb}||_\infty  ||\e^{(s-k)b}\bigl( \mathcal E(z+h) - \mathcal E(z) \bigr)  || \\ & + ||z \e^{kb}||_\infty  ||\e^{(s-k)b}\bigl( \mathcal E(z+h) - \mathcal E(z) - \bigl(\mathrm D\mathcal E(z) h \bigr) \bigr)  ||	= o(||h\e^{kb}||_\infty\bigr). 
	\end{align*}
 	Hence $\rho$ is holomorphic on $||z\e^{kb}||_\infty < 1$ with values in $Y_{sb}$ for $s+1 \leq k$. Furthermore, $\rho$  is bounded by $1$ because of~\eqref{oboe}. The result is the extremal case $s=k-1$.

 	The map $\zeta$ is treated in a completely analogous way.
 	We start from $\zeta(q) = \nu(q) T_q^\circ(\nu)$. Since we assume that  $m$ satisfies condition~\eqref{suff1}, we know from Theorem~\ref{thm:main1} that 
 	\be \label{eq:t0bound}
 		1+\sum_{n=1}^\infty \frac{1}{n!}\int_{\mathbb X^n}|t_n(q;x_1,\ldots,x_n)|\e^{b(x_1)+\cdots + b(x_n)} m^n(\dd \vect x) \leq \e^{b(q)}. 
 	\ee
	We can now repeat the reasoning for $\rho[z]$, substituting $\nu$ for $z$, $T_q^\circ(q;\nu)$ for $E(q;z)$, and the bound~\eqref{eq:t0bound} for~\eqref{eq:m0bound}. 
\end{proof} 

\subsection{An equivalent fixed point equation} 

In the proof of Lemma~\ref{lem:tada} in Section~\ref{sec:virial} we need another characterization of the coefficients $t_n(q;x_1,\ldots,x_n)$. 

\begin{lemma}  \label{lem:tree-sol2}
	The family $(T_q^\circ)_{q\in \mathbb X}$ from Lemma~\ref{lem:treesolution} is the unique family of formal power series that solves 
	\be \label{tree-eq2} \tag{$\mathsf{FP}'$}
		1+ \sum_{n=1}^\infty \frac{1}{n!}\int_{\mathbb X^n}  t_n(q;x_1,\ldots,x_n) \prod_{i=1}^n \e^{- A(x_i;z)} z(\dd x_1)\cdots z(\dd x_n) =\e^{A(q;z)}.
	\ee
\end{lemma}

\noindent Eq.~\eqref{tree-eq2} reflects that $T_q^\circ(\rho[z]) = \exp( A(q;z))$ while the fixed point equation~\eqref{tree-eq}, defining $(T_q^\circ)_{q\in \mathbb X}$, reflects that $T_q^\circ (\nu) = \exp(A(q;\nu T_q^\circ(\nu)))$ because $\rho(\xi(\nu))= \nu$. 

\begin{proof} 
	Let us write $\tilde t_n$ instead of $t_n$ as long as we do not know that the family from Lemma~\ref{lem:treesolution} satisfies~\eqref{tree-eq2}. For the existence and uniqueness of a solution $(\tilde T_q^\circ)_{q\in \mathbb X}$ to~\eqref{tree-eq2},  we note that Eq.~\eqref{tree-eq2} translates into a triangular system of equations for the coefficients $\tilde t_n$. The details are similar to the proof of Lemma~\ref{lem:treesolution} and therefore omitted. 
	
%\tkk{ Next we show $\tilde T_q^\circ = T_q^\circ$. The intuitive reasoning is as follows. Let $\tilde \zeta[\nu](\dd q) := \nu(\dd q) \tilde T_q^\circ[\nu]$. Then 
% $$
% 	\tilde \zeta[\rho[z]] (\dd q)= \rho[z](\dd q) \tilde T_q^\circ[\rho[z]] =  \Bigl( z(\dd q) \e^{-A(q;z)}\Bigr) \e^{A(q;z)}  = z(\dd q)
% $$
% hence $\tilde \zeta$ is a left inverse of $\rho$. By the same reasoning based on~\eqref{tree-eq2}, $\rho[\zeta[\nu]]= \nu$ hence $\zeta$ is a right inverse of $\rho$. But left and right inverse are equal, since 
% $$
% 	\zeta =\mathrm {id} \circ \zeta =(\tilde \zeta\circ \rho) \circ \zeta = \tilde \zeta\circ (\rho\circ \zeta) = \tilde \zeta \circ \mathrm{id}= \tilde \zeta.
%$$
%Thus we should have $\zeta = \tilde \zeta$ and $T_q^\circ = \tilde T_q^\circ$. 
%
%The intuitive argument can be made rigorous by introducing measure-valued formal power series, but we choose to proceed more directly. }{}\comment{I do not what this add beyond the proof and the heuristic explenation at the start, so I suggest to cut it}

 We start from~\eqref{tree-eq2}, written for $\tilde t_n$'s instead of $t_n$'s, and insert $z(\dd q) = \nu(\dd q) T_q^\circ(\nu)$ on both sides. This insertion corresponds precisely to the second notion of composition discussed in Appendix~\ref{ApFormal}, see Eq.~\eqref{fcomp2}, and in particular it is a well-defined operation on formal power series.  The composition yields two formal power series in $\nu$, one for the left and one for the right side, called $L$ and $R$ respectively, and of course we must have $L(q;\nu) = R(q;\nu)$. On the right side we get, by~\eqref{tree-eq}, 
$$
	R(q;\nu) = \exp(A(q;\nu T^\circ(\nu))  = T_q^\circ(\nu).
$$
On the left side we have
\begin{align*}
	L(q;\nu) 
			& = 1+ \sum_{n=1}^\infty \frac{1}{n!} \int_{\mathbb X^n}\tilde t_n(q;x_1,\ldots,x_n) \prod_{i=1}^n \e^{- A(x_i;\nu T^\circ(\nu))} \prod_{i=1}^n \Bigl( T_{x_i}^\circ(\nu) \nu(\dd x_i)\Bigr)  \\
			& = 1+ \sum_{n=1}^\infty \frac{1}{n!} \int_{\mathbb X^n}\tilde t_n(q;x_1,\ldots,x_n) \prod_{i=1}^n \Bigl( \e^{- A(x_i;\nu T^\circ(\nu))} T_{x_i}^\circ(\nu) \Bigr)\nu (\dd x_1)\cdots \nu(\dd x_n). 
\end{align*} 
The product inside the integral is equal to $1$ because of~\eqref{tree-eq}, therefore $L(q;\nu) = \tilde T_q^\circ(\nu)$ and we conclude from $L=R$ that $\tilde T_q^\circ(\nu) = T_q^\circ(\nu)$. In particular, $(T_q^\circ)_{q\in \mathbb X}$ solves~\eqref{tree-eq2}. 
\end{proof}

\section{Virial expansion. Density functional} \label{sec:virial} 

In this section we are consider functions $A$ of a special form, cf.~\eqref{A1} below, which are coming from a system of objects interacting via  a pair potential.

Let $V:\mathbb X\times \mathbb X\to \mathbb R\cup \{\infty\}$ be a measurable pair potential ($V(x,y) = V(y,x)$). We assume that for some measurable function $B:\mathbb X\to [0,\infty)$, we have the stability condition 
\begin{equation}\label{stability}
	\sum_{1\leq i<j\leq n} V(x_i, x_j)\geq -\sum_{i=1}^n B(x_i),
\end{equation}
for all $n\geq 2$ and $x_1,\ldots, x_n\in \mathbb X$. 
In addition, we also assume that for all $x\in \mathbb X$  and some function $B^*:\mathbb X\to \R_+$ we have
\begin{equation}\label{stabextra}
	\inf_{y\in \mathbb X} V(x,y)\geq - B^*(x).
\end{equation}
Define 
$$
	H_n(x_1,\ldots, x_n):= \sum_{1\leq i<j\leq n} V(x_i, x_j),
$$
for $n\geq 2$ and $H_0=0$, $H_1=0$. Let us introduce for the next few calculation up to \eqref{eq:difflog} an extra assumption on  $z \in \mathfrak M_\C(\mathbb X,\mathcal X)$, namely
\be \label{eq:fivo}
	\int_\mathbb X \e^{\beta B(x)}| z|(\dd x) < \infty.
\ee
In the case that $\mathbb{X} \subset \mathbb{R}^d$ and $V$, $z$ respectively, is a translation invariant function, measure respectively, then the above condition means that the volume of $\mathbb{X}$ with respect to the Lebesgue measure is finite. Hence we say that we are in the ``finite volume" case. We will point out which formulas also hold in the ``infinite volume" case. 

The grand-canonical partition function at activity $z$ and inverse temperature $\beta>0$ is 
\be\label{partition}
	\Xi(\beta, z)=1+ \sum_{n=1}\frac{1}{n!}	\int_{\mathbb X ^n} \e^{-\beta  H_n(\vect x)} z^n(\dd \vect x).
\ee
Condition~\eqref{eq:fivo} ensures that $\Xi(\beta,z)$ is finite. 
The one-particle density is
\be \label{eq:density}
	\rho[z](\dd q) = \rho(\dd q;z):= \frac{1}{\Xi(\beta,z)} \Biggl(1+ \sum_{n=1}^\infty \frac{1}{n!} \int_{\mathbb X^n}\e^{- \beta H_{n+1}(q,x_1,\ldots, x_n)} z^n(\dd \vect x) \Biggr) z(\dd q). 
\ee
Notice 
\be\label{eq:difflog}
	\rho(\dd q;z) =\Bigl( \frac{\delta}{\delta z(q)} \log \Xi(\beta, z)\Bigr) z(\dd q),
\ee
see Eqs.~\eqref{fvardev} and~\eqref{fvardev2} in Appendix~\ref{ApFormal} applied to $\log \left( 1 + \left(\Xi(\beta, z)-1\right)\right)$.
We bring the expression for $\rho$ into the form~\eqref{eq:rhoa}. This allows us to extend the definition~\eqref{eq:density} to activities that do not satisfy the finite-volume condition~\eqref{eq:fivo}.
Set 
\be
	f(x,y):=\e^{-\beta V(x,y)}-1,\quad \bar f(x,y):= 1 - \e^{-\beta|V(x,y)|}.
\ee	
	 Let $\mathcal C_n$ be the set of connected graphs  $g$ with vertex set $[n]=\{ 1 , \dots , n\}$, and $E(g)$ the edge set of a graph $g=([n], E(g))$ and 
\begin{equation}\label{A1}
	A_n(q; x_1,\ldots,x_n):=-\left[\prod_{j=1}^n (1+f(q,x_j))-1\right]
	\sum_{g\in\mathcal C_n}\prod_{\{i,j\}\in E(g)}f(x_i, x_j).
\end{equation}
The aim of the section is to use the result of the previous section for this particular $A_n$. Furthermore, define the well-known Ursell functions 
	\be \label{eq:ursell}
		\varphi_n^\mathsf T(x_1,\ldots,x_n):= 	\sum_{g\in\mathcal C_n}\prod_{\{i,j\}\in E(g)}f(x_i, x_j).
	\ee
	Let us recall some known results.
	
\begin{lemma} \label{lem:prelim}
	Let $A_n(q;x_1,\ldots,x_n)$ be as in~\eqref{A1} and define $A(q;z)$ as in~\eqref{Adef}. Let $z\in \mathfrak M_\C$ satisfy only
	\be \label{PU}
		\int_\mathbb X \bar f(x,y)\, \e^{a(y)+ \beta B(y)}|z|(\dd y)\leq a(x)
	\ee
	for some weight function $a:\mathbb X\to \R_+$ and all $x\in \mathbb X$. Then $z$ is in the domain of convergence $\mathscr D(A)$. 
	
	If in addition $z$ satisfies the finite-volume condition~\eqref{eq:fivo}, then the density $\rho(\dd q;z)$ defined in~\eqref{eq:density} is equal to $\exp( - A(q;z)) z(\dd q)$, moreover
	\begin{align*}
		\log \Xi(\beta,z) & = \sum_{n=1}^\infty \frac{1}{n!}\int_{\mathbb X^n}  \varphi_{n}^\mathsf T(x_1,\ldots,x_n) z^n(\dd \vect x),\\
		\rho(\dd q;z) & = z(\dd q)\Biggl( 1+ \sum_{n=1}^\infty \frac{1}{n!}\int_{\mathbb X^n} \varphi_{n+1}^\mathsf T(q,x_1,\ldots,x_n) z^n(\dd \vect x)\Biggr)
	\end{align*} 
	with absolutely convergent integrals and series.
\end{lemma}

\noindent The lemma follows from the tree-graph inequality due to~\cite{procacci-yuhjtman2017} and additional combinatorial considerations, compare for example~\cite[Eq. (4.17)]{jttu2014}. The details are similar to aspects of the proof of Lemma~\ref{lem:abkey} and therefore omitted.

\begin{definition}\label{def:rho1}
For  activities $z$ that satisfy~\eqref{PU} but not necessarily the condition~\eqref{eq:fivo},  we adopt the equality $\rho(\dd q;z) = z(\dd q) \exp( - A(q;z))$ as the definition of the density. 
\end{definition} 

\begin{remark}[Physical interpretation of $A(q;z)$]
	Let $W(q;x_1,\ldots,x_n) := \sum_{i=1}^n V(q,x_i)$ be the total interaction of a particle at $q$ with the particles $x_1,\ldots, x_n$. 
	By~\eqref{eq:density} and Lemma~\ref{lem:prelim}, we have 
	$$
		 \frac 1\beta A(q;z) = - \frac 1\beta  \log \Bigl \la \e^{- \beta W(q;x_1,\ldots,x_n)  }\Bigr \ra ,
	$$
	where $\la \cdot \ra$ denotes the expectation with respect to the grand-canonical Gibbs measure. Thus $\frac1\beta A(q;z)$ is the excess free energy for a test particle pinned at the location $q$. 
\end{remark} 

Let $\mathcal B_n\subset \mathcal C_n$ be the set of bi-connected graphs, i.e., graphs that stay connected upon removal of a single vertex. Define 
\be \label{eq:Dn}
	D_n(x_1,\ldots, x_n) := \sum_{g\in \mathcal B_n} \prod_{\{i,j\}\in E(g)} f(x_i,x_j). 
\ee
We want to invert the map $z\mapsto \rho[z]$ and express the inverse with bi-connected graphs.  Before that we derive a convergent result for power series with  coefficients given by bi-connected graphs.

\begin{theorem} \label{thm:virmain1}
	Let $\nu\in \mathfrak M_\C$. Suppose there exist functions 	$a,b:\mathbb X\to \R_+$ with $a\leq b$ on $\mathbb X$ 	such that 
	\be \label{suffsuff} \tag{$\mathsf S_{a,b}$}
		\int_\mathbb X \bar f(x,y)\, \e^{a(y)+b(y)+ \beta B(y)+\beta B^*(y)}|\nu|(\dd y)\leq a(x),
	\ee
	for all $x\in \mathbb X$. Then 
	\be \tag{$\mathsf M_b$} \label{virMb}
		\sum_{n=1}^\infty \frac{1}{n!}\int_{\mathbb X^n} \bigl| D_{n+1}(q,x_1,\ldots, x_n) \bigr|\, |\nu|(\dd x_1) \cdots |\nu|(\dd x_n) \leq  b(q)
	\ee
	for all $q\in \mathbb X$. 
\end{theorem}

\noindent Define $\mathsf V_b$ by
$$
	\mathsf V_b = \bigl\{ \nu \in \mathfrak M_\C\mid \exists a:\mathbb X\to \R_+:\, a\leq b,\ \nu\text{ satisfies~\eqref{suffsuff}}\bigr\}. 
$$

\begin{theorem} \label{thm:virmain2}
	There is a set $\mathsf U_b\subset \mathscr D(A) \subset \mathfrak M_\C$ such that $z\mapsto \rho[z]$ is a bijection from $\mathsf U_b$ onto $\mathsf V_b$, and for every $z\in \mathsf U_b$, $\nu\in \mathsf V_b$, we have $\rho[z]=\nu$ if and only if 
	\be \label{eq:virmain2a} 
		z(\dd q ) = \nu(\dd q) \exp\Biggl(-  \sum_{n=1}^\infty \frac{1}{n!}\int_{\mathbb X^n}  D_{n+1}(q,x_1,\ldots, x_n) \nu(\dd x_1) \cdots \nu(\dd x_n) \Biggr),
	\ee
where the latter converge in the sense that  \eqref{virMb} holds.

If $z\in \mathfrak M_\C$ fulfills  \eqref{suffsuff} for some $a\leq b$ and $\e^{a}|z|\in \mathsf V_b$ for the same functions $a$ and $b$, 
%satisfies 
%$\e^{a}|z|\in \mathsf V_b$ for some $a\leq b$ with $|z|$ fulfills 
then  $\rho[z]\in \mathsf V_b$ and hence $z \in \mathsf U_b$. 

If instead the following conditions including also a ``finite volume condition" holds, 
	\be \label{dissymmetry-condition:b}
		\int_{\mathbb X}\bar f(x,y)\, \e^{a(y)+ \beta B(y)} |z|(\dd y) \leq a(x),\quad 	\e^{a+\beta B}|z|\in \mathsf V_b, \quad \int_{\mathbb X}(1+ b(q)) \e^{a(q)+ \beta B(q)}|z|(\dd q)<\infty,
	\ee
	then also
	\be \label{eq:virmain2b}
		\log \Xi(\beta,z) = \int_\mathbb X\rho(\dd x_1;z) 
			- \sum_{n=2}^\infty \frac{1}{n!}\int_{\mathbb X^n} (n-1) D_n(x_1,\ldots,x_n) \prod_{i=1}^n \rho(\dd x_i;z).
	\ee
\end{theorem}

\noindent The condition~$\e^a z\in \mathsf V_b$ is a condition directly in terms of $z$ which is sufficient to guarantee that $z \in \mathsf  U_b$. Recall that $\mathsf  U_b$ was just defined indirectly as the image of $\zeta$.

Formula \eqref{eq:virmain2b} does not make any sense in the ``infinite volume case" even if we consider the translation invariant case as discussed below \eqref{eq:fivo}. In this case, though, the right hand side is proportional to the  volume of $\mathbb{X}$, up to boundary errors. Hence, $\log \Xi(\beta,z)$ divided by the volume has a well defined limit. 

%\blue{[Condition $\e^{a+\beta B}|z|\in \mathsf V_b$ means that there exists some function $a'$ not necessarily equal to $a$ such that 
%$$
%\int_\mathbb X \bar f(x,y)\, \e^{a'(y)+b(y)+ \beta B(y)+\beta B^*(y)} \, \e^{a(y) + \beta B(y)} |z|(\dd y)\leq a'(x).
%$$
%It does not imply the newly added blue part  in condition~\eqref{dissymmetry-condition}. ]
%} 

For the definition of the free energy, we fix a reference measure $m(\dd x)$ on $\mathbb X$ (for example, the Lebesgue measure on $\R^d$). The (grand-canonical) free energy $\mathcal F_{\mathrm{GC}}[\nu]$ of a given density profile $\nu \in \mathfrak M$ is defined via the Legendre transform of $\log \Xi(z)$ as 
\begin{equation}\label{GCFELT}
	\beta \mathcal F_{\mathrm{GC}}[\nu]:= \sup_{z}\Bigl(\int_\mathbb X \log \frac{\dd z}{\dd m}(x) \nu(\dd x) - \log \Xi(z)\Bigr)
\end{equation}
with $\frac{\dd z}{\dd m}$ the Radon-Nikod{\'y}m derivative of $z$ with respect to the reference measure $m$. The supremum in~\eqref{GCFELT} is over all non-negative measures $z\in \mathfrak M$ that are absolutely continuous with respect to $m$ and such that the integral with the logarithm is absolutely convergent. 

\begin{theorem}\label{thm:GCFE}
	Assume that $\nu \in \mathsf V_b\cap \mathfrak M$ is absolutely continuous with respect to $m$ and satisfies 
	\be \label{eq:nuassumption}
		\int_{\mathbb X}(1+b(q)) \nu(\dd q)<\infty,\quad 
		\int_\mathbb X\Bigl| \log \frac{\dd \nu}{\dd m}\Bigr|\dd \nu <\infty,\quad \int_\mathbb X \e^{\beta B+ b} \dd \nu <\infty, 
	\ee
	then 
	\begin{equation} \label{eq:free-energy}
		\beta \mathcal F_{\mathrm{GC}}[\nu]
			 = \int_\mathbb X   \bigl[\log \frac{\dd \nu}{\dd m}(x) - 1\bigr]\nu(\dd x)  - \sum_{n=2}^\infty \frac{1}{n!} \int_{\mathbb X^n}D_n\bigl(x_1,\ldots, x_n\bigr) \nu^n(\dd \vect x)
	\end{equation} 
	with absolutely convergent integrals and sum. 
\end{theorem} 

\noindent Let us first check that condition~\eqref{suffsuff} is sufficient for the convergence of $A(q;z)$.

\begin{lemma} \label{lem:abkey}
	If $\nu$ satisfies condition~\eqref{suffsuff} for some $a,b:\mathbb X\to \R_+$ with $a\leq b$, then $\nu$ satisfies $|A(q;\nu)| \leq a(q)$ and in particular condition~\eqref{suff1}, where $A$ is defined as in \eqref{Adef} with $A_n$ given by~\eqref{A1}. 
\end{lemma} 

\begin{proof}
	Set 
	$$
		\mathcal R(y;\mu):= 1 +  \sum_{m=1}^\infty \frac{1}{m!}\int_{\mathbb X^m} \bigl|	\varphi_{m+1}^\mathsf T(y,x_1,\ldots,x_m)\bigr|\, |\mu|^m(\dd \vect x).
	$$
The first factor in \eqref{A1} can be bounded as follows
	\be\label{usualmagicformula}
		\left|\prod_{j=1}^n (1+f(q,x_j))-1\right|
	\leq \e^{\beta \sum_{j=1}^n B^*(x_j)}\sum_{i=1}^n \bar f(q,x_i) 
	\ee
	Indeed, this follows by induction in $n$ using 
	\[
	\left| \prod_{j=1}^{n+1} (1+f(q,x_j))-1 \right| \leq |1 + f(q, x_{n+1}) | \left| \prod_{j=1}^{n} (1+f(q,x_j))-1 \right| + |f(q, x_{n+1}) |
	\]
	 and using that $| e^{-u} -1 | \leq e^{\max\{-u,0\}} ( 1- e^{-|u|})$.
	Using this bound, we get 
	\begin{align}
		& \sum_{n=1}^\infty \frac{1}{n!}\int_{\mathbb X ^n }|A_n(q;x_1,\ldots,x_n)|\e^{\sum_{j=1}^n b(x_j)} |\nu|(\dd x_1)\cdots |\nu|(\dd x_n) \notag \\
		&\quad \leq \sum_{n=1}^\infty \frac{1}{n!}\int_{\mathbb X^n} \e^{ \sum_{j=1}^n (\beta B^*(x_j)+b(x_j))}\sum_{i=1}^n \bar f(q,x_i)\, 	\bigl|	\varphi_n^\mathsf T(x_1,\ldots,x_n)\bigr|\, |\nu|(\dd x_1)\cdots |\nu|(\dd x_n)\notag \\
		& \quad = \int_{\mathbb X} \bar f(q,y)\,   \mathcal R\bigl(y; \e^{\beta B^*+b} |\nu|\bigr)\,  \e^{\beta B^*(y)+ b(y)}\, |\nu|\, (\dd y). \label{unours}
	\end{align} 
	In order to bound $\mathcal R(q;\e^{\beta B^*+b}\nu)$, we use a recent tree-graph inequality due to Procacci and Yuhjtman~\cite{procacci-yuhjtman2017} in the form presented in~\cite{ueltschi2017}.  Then 
	$$
		\bigl|\varphi_n^\mathsf T(x_1,\ldots,x_n)\bigr| 
			\leq \e^{\beta B(x_1)+\cdots +\beta B(x_n)} \sum_{T\in \mathcal T_n}\prod_{\{i,j\}\in E(T)} \bar f(x_i,x_j), 
	$$
	with $\mathcal T_n\subset \mathcal C_n$ the set of trees with vertex set $[n]$. As a consequence, if a non-negative measure $\mu$ satisfies 
	\be \label{fbarcond}
		\int_\mathbb X \bar f(q,y) \e^{a(y)+\beta B(y)} \mu(\dd y) \leq a(q)
	\ee
	for all $q\in \mathbb X$, then 
	\be \label{pyrq}
		\mathcal R(q;\mu) \leq \e^{a(q)+\beta B(q)}.
	\ee
	The inductive proof of~\eqref{pyrq} is similar to the proof of \cite[Theorem 2.1]{poghosyan-ueltschi2009}  and therefore omitted. 
Condition~\eqref{suffsuff} implies that  $\mu:= \exp(\beta B^*+ b)|\nu|$ satisfies
	$$
		\int_{\mathbb X} \bar f(x,y)\e^{a(y)+ \beta B(y)}\mu(\dd y) 
			= \int_{\mathbb X} \bar f(x,y)\, \e^{a(y) +  \beta B(y) + \beta B^*(y) + b(y)} \, |\nu|(\dd y) \leq a(y).
	$$
	Hence ~\eqref{fbarcond} and~\eqref{pyrq} hold true, and we can further bound~\eqref{unours} by 
	\begin{multline*} 
		 \int_{\mathbb X} \bar f(q,y)   \mathcal R\bigl(y; \e^{\beta B^*+b} |\nu|\bigr)\,  \e^{\beta B^*(y)+ b(y)}\, |\nu|\, (\dd y) 	\\	 \leq  \int_{\mathbb X} \bar f(q,y)\,  \e^{ a(y)+ \beta B(y)+\beta B^*(y) + b(y)}\, |\nu|\, (\dd y) \leq a(q) \leq b(q)
	\end{multline*}
	which completes the proof. 
\end{proof} 

\noindent Next let us relate the coefficients of $A(q;z)$ with bi-connected graphs.

\begin{lemma} \label{lem:biconnected}
	The formal power series $A(q;z)$ with coefficients~\eqref{A1} satisfies 
	\be \label{eq:biconnected}
		- A(q;z) = \sum_{n=1}^\infty \frac{1}{n!}\int_{\mathbb X^n} D_{n+1}(q,x_1,\ldots,x_n)\prod_{i=1}^n \e^{- A(x_i;z)} z^n(\dd \vect x).
	\ee
\end{lemma}

\begin{proof} 
	The lemma follows from well-known identities for connected and bi-connected graphs, see for example~\cite{stell1964, leroux2004, faris2012biconnected, morita-hiroike3}, we sketch the argument for the reader's convenience. 
	If $J\subset \N$ is a finite non-empty set, consider the following classes of graphs with vertex set $J\cup \{0\}$: 
	\begin{itemize}
		\item $\mathcal C^\circ(J)$, the connected graphs on $J\cup \{0\}$; 
		\item $\mathcal B^\circ(J)$, the biconnected graphs on $J\cup\{0\}$;
		\item $\mathcal A^\circ(J)$, the connected graphs that stay connected when removing $0$ and the incident edges (equivalently, the connected graphs for which $0$ is not an articulation point). 
	\end{itemize} 	
	If $g$ is a graph with vertex set $J\cup\{0\}$, define $w(g;(x_i)_{i\in J\cup\{0\}}) = \prod_{\{i,j\}\in E(g)} f(x_i,x_j)$. Then 
	\be \label{eq:anconn}
		- A_n(q;x_1,\ldots,x_n) = \sum_{g\in \mathcal A^\circ([n])} w(g;q,x_1,\ldots,x_n). 
	\ee
	In view of~\eqref{eq:formal-exponential}, setting $x_0=q$, the coefficients of $\exp( - A(q;z))$ are given by 
	\begin{align}
		\mathcal E_n(q;x_1,\ldots, x_n) &=\sum_{m=1}^n \sum_{\{J_1,\ldots, J_m\}\in \mathcal P_n} \prod_{k=1}^n\Biggl( \sum_{g_k\in \mathcal A^\circ(J_k)} w(g_k; (x_j)_{j\in J_k\cup \{0\}}) \Biggr)\notag  \\
		&= \sum_{g \in \mathcal C^\circ([n])} w(g;q,x_1,\ldots,x_n). \label{eq:enconn} 
	\end{align} 
	By  Eq.~\eqref{fcomp2}, the right-hand side of~\eqref{eq:biconnected} is a power series $F(q;z)$ with coefficients 
	$$
		F_n(q;x_1,\ldots, x_n)  = \sum_{m=1}^n \sum_{\substack{L\subset [n]  \\ \#L =m}}  D_{m+1}\bigl((x_j)_{j\in J\cup\{0\}}\bigr) \sum_{\substack{(J_\ell)_{\ell\in L }:\\  \dot \cup_{\ell \in L} J_\ell = [n]\setminus L}} \prod_{\ell \in L} \mathcal E_{\# J_\ell}\bigl(x_\ell; (x_j)_{j\in  J_\ell}\bigr).
	$$
	Eq.~\eqref{eq:enconn} allows us to rewrite $F_n(q;x_1,\ldots,x_n)$ as a sum over tuples $(m,g_0,g_1,\ldots, g_m)$ consisting of an integer $m\in \{1,\ldots,n\}$ and graphs $g_0\in \mathcal B^\circ(L)$, $g_\ell \in \mathcal C^\circ(J_\ell)$ where $\# L=m$ and $L,J_1,\ldots,J_\ell$ form a partition of $[n]$ with $J_\ell = \varnothing$ allowed. 
	Given such a tuple $(m,g_0,g_1,\ldots,g_m)$, a new graph $g$ is defined by gluing each $g_\ell$ to $g_0$ at the vertex $\ell$ (the vertex $\ell$ is identified with root $0$ of $g_\ell$). Precisely, $\{i,j\}$ is an edge of $g$ if and only if:
	\begin{itemize}
		\item either $i,j \in L$ and $\{i,j\}\in E(g_0)$, 
		\item or for some $\ell\in L$ we have $i,j\in J_\ell$ and $\{i,j\}\in E(g_\ell)$,
		\item or for some $\ell\in L$ we have $i=\ell$ and $j\in J_\ell$ (or vice-versa) and $\{0,j\}\in E(g_\ell)$.
	\end{itemize}
	In the new graph $g$, each of the vertices $\ell\in L$ is an articulation point (that is upon the removal of $\ell$ and the edges incident to $\ell$ the graph $g$ has a connect component which does not contain $0$. However, note that  there can be other articulation points inside the $J_\ell$'s!), and the support $J_\ell$ of the graph $g_\ell$ consists of those vertices $j\in [n]$ for which every path connecting $j$ to $0$ has to pass through $\ell$. The weight of the new graph is equal to the product of the weights of the $g_\ell$'s.
	
	The rule $(m,g_1,\ldots,g_m)\mapsto g$ defines a one-to-one correspondence between the tuples under consideration and graphs $g\in \mathcal A^\circ([n])$, and the weights are multiplicative. One deduces that $F_n(q;x_1,\ldots,x_n)$ is given by a sum over graphs $g\in \mathcal A^\circ ([n])$ and weights as in~\eqref{eq:anconn}, therefore~\eqref{eq:biconnected} holds true. 
\end{proof}

\noindent As a consequence we can identify the coefficients of $T_q^\circ(\nu)$.

\begin{lemma}\label{lem:tada}
	For $A_n(q;x_1,\ldots,x_n)$ given by~\eqref{A1}, the family $(T_q^\circ)_{q\in\mathbb X}$ from Lemma~\ref{lem:treesolution} is given by 
	\be \label{tada}
		T_q^\circ(\nu) =  \exp\Biggl( - \sum_{n=1}^\infty \frac{1}{n!} \int_{\mathbb X^n} D_{n+1}(q,x_1,\ldots,x_n) \nu(\dd x_1)\cdots \nu(\dd x_n) \Biggr).
	\ee
\end{lemma} 

\begin{proof}
	Lemma~\ref{lem:biconnected} yields
	\be
		\exp\Biggl( - \sum_{n=1}^\infty \frac{1}{n!} \int_{\mathbb X^n} D_{n+1}(q,x_1,\ldots,x_n) \prod_{i=1}^n \e^{- A(x_i;z)} z^n(\dd \vect x)\Biggr) = 	\e^{A(q;z)}.  
	\ee
	Hence the right-hand side of~\eqref{tada} solves the fixed point equation~\eqref{tree-eq2} as considered in Lemma~\ref{lem:tree-sol2},furthermore the lemmas yields that, as the solution of the fixed point equation, the right hand side must be equal to the family $(T_q^\circ)_{q\in\mathbb X}$ from Lemma~\ref{lem:treesolution}. 
\end{proof} 

\begin{proof}[Proof of Theorem~\ref{thm:virmain1}]
	If $\nu$ satisfies~\eqref{suffsuff}, then by Lemma~\ref{lem:abkey} it also satisfies~\eqref{suff1}. However, by Theorem~\ref{thm:main1}, it follows that~\eqref{mainest1} holds true as well, in particular $T_q^\circ(\nu)$ is absolutely convergent and $|T_q^\circ(\nu)|\leq \exp(b(q))$. Combining Eqs.~\eqref{tada} and~\eqref{tree-eq} we get 
	$$
		- \sum_{n=1}^\infty \frac{1}{n!} \int_{\mathbb X^n} D_{n+1}(q,x_1,\ldots,x_n) \nu^n(\dd \vect x)
			= \sum_{n=1}^\infty \frac{1}{n!}\int_{\mathbb X^n} A_n(q;x_1,\ldots,x_n) \prod_{i=1}^n T_{x_i}^\circ(\nu) \nu^n(\dd \vect x). 
	$$
as formal power series, that means, that the coefficients of the series coincides. If we take the absolute value of the coefficients and we reconstruct the right hand side of the above equality one gets that 
\begin{multline*}
				\sum_{n=1}^\infty \frac{1}{n!} \int_{\mathbb X^n} \bigl|D_{n+1}(q,x_1,\ldots,x_n)  \bigr|\, |\nu|^n(\dd \vect x)
=  \sum_{n=1}^\infty \frac{1}{n!}\int_{\mathbb X^n} \bigr|A_n(q;x_1,\ldots,x_n)\bigr| \prod_{i=1}^n \bigl|S_{x_i}(\nu)\bigr|\, |\nu|^n(\dd \vect x),
	\end{multline*}
where we define
\bes
	S_q(\nu) :=1+ \sum_{n=1}^\infty \frac{1}{n!} \int_{\mathbb X^n}\left| t_{n}(q;x_1,\ldots, x_n)\right| \   |\nu| (\dd x_1)\cdots  |\nu|(\dd x_n).
\ees	
	The right-hand side is bounded by $b(q)$ because of~\eqref{mainest1} and~\eqref{suff1}. 
\end{proof} 

\begin{proof}[Proof of Theorem~\ref{thm:virmain2}]
		Let $\zeta[\nu] (\dd q) = \zeta(\dd q;\nu)= \nu(\dd q) T_q^\circ(\nu)$ as in~\eqref{zetadef}. Set $\mathsf U_b:= \zeta[\mathsf V_b]$. By Lemma~\ref{lem:abkey}, we know that $\mathsf V_b\subset \mathscr V_b$ hence Theorem~\ref{thm:main2} guarantees $\mathsf U_b\subset \mathscr U_b \subset \mathscr D(A)$. 
 It follows from Theorem~\ref{thm:main2} that $\rho$ is a bijection from $\mathsf U_b$ onto $\mathsf V_b$ with inverse $\zeta$, hence $\rho[z] = \nu$ if and only if $z(\dd q) = \nu(\dd q) T_q^\circ (\nu)$. We insert the formula~\eqref{tada} from Lemma~\ref{lem:tada} for $T_q^\circ(\nu)$ and obtain~\eqref{eq:virmain2a}. 

Let $z\in \mathfrak M_\mathbb C$ satisfy \eqref{suffsuff} and $\e^{a}|z|\in \mathsf V_b$, then  $|A(q;z)|\leq a(q)$ by Lemma~\ref{lem:abkey}. By Definition~\ref{def:rho1}, the density is given by $\rho(\dd q;z) = z(\dd q) \e^{-A(q;z)}$ which is bounded by $|\rho|(\dd q; z) \leq | z|(\dd q) \e^{a(q)} \in \mathsf V_b$. 
 As $\zeta[\rho[z]]=z$ in the sense of formal power series and $\zeta$ is a convergent on $ \mathsf  V_b$, it remains to show that the composition is also convergent. For that we do not only need that $\rho[z] \in \mathsf V_b$ but also that all the interchanges are allowed, that is, an estimate in terms of $|z|$ and $|A_n|$, namely~\eqref{suff1}. Therefore, we finally get $z \in  \mathsf  U_b$.

%
%
%we get that
%\[
% \int_\mathbb X \bar f(x,y)\, \e^{a(y)+b(y)+ \beta B(y)+\beta B^*(y)}|\rho(\dd y;z) | \leq   \int_\mathbb X \bar f(x,y)\, \e^{a(y)+b(y)+ \beta B(y)+\beta B^*(y)} e^{a(y)}|z(\dd y)| \leq a(x)
%\]
%and hence $\rho[z] \in V_b$. 

	As an equality of formal power series, Eq.~\eqref{eq:virmain2b} follows from the dissymmetry theorem for connected and biconnected graphs and power series manipulations similar to the proof of Lemma~\ref{lem:biconnected}.  Precisely, we have the following identity 
\begin{multline} \label{eq:dissymmetry} 
		\varphi_n^\mathsf T(x_1,\ldots,x_n) = n \varphi_n^\mathsf T(x_1,\ldots,x_n) \\
		 -\sum_{m=2}^n (m-1) \sum_{\substack{L\subset[n]\\ \#L= m}} D_{m}\bigl( (x_\ell)_{\ell\in L} \bigr) \sum_{\substack{(J_\ell)_{\ell\in L}:\\ \dot \cup J_\ell = J}} \prod_{\ell \in L} \varphi_{\#J_\ell +1}^\mathsf T\bigl( (x_j)_{j\in J_\ell\cup \{\ell\}}\bigr).   	
\end{multline} 
The proof  of~\eqref{eq:dissymmetry} is easily adapted from~\cite[Theorem 3.1]{jttu2014} or~\cite{leroux2004} and therefore omitted. 
The first part of condition~\eqref{dissymmetry-condition:b} is		 condition~\eqref{PU} from Lemma~\ref{lem:prelim}, we have established \eqref{eq:virmain2b} in the sense of formal power series. Next, we check absolute convergence of the power series associated with the terms in Eq.~\eqref{eq:dissymmetry}. 

Let us consider~\eqref{eq:dissymmetry} term by term starting from the left. 
Consider
\be \label{defR}
	\mathcal R(q;|z|)= 1+ \sum_{n=1}^\infty \frac{1}{n!}\int_{\mathbb X^n} \bigl|\varphi_{n+1}^\mathsf T(q,x_1,\ldots,x_n)\bigr| |z|^n(\dd \vect x). 
\ee
The  first part of condition~\eqref{dissymmetry-condition:b} is the same as condition~\eqref{fbarcond} with $|z|$ instead of $\mu$, so we may apply the bound~\eqref{pyrq} and get that
\be \label{dis2}
	\mathcal R(q;|z|) \leq \e^{a(q)+ \beta B(q)}.
\ee
Hence the formal power series for $\log \Xi(\beta,z)$ is converging exactly in the sense that \eqref{defR} is finite.  

Next, by~\eqref{dis2} and condition~\eqref{dissymmetry-condition:b}, we also have
\be \label{dis3}
	\sum_{n=1}^\infty \frac{1}{n!}\int_{\mathbb X^n} \bigl| n \varphi_n^\mathsf T(x_1,\ldots,x_n)\bigr|\, |z|^n(\dd \vect x)  
	\leq \int_\mathbb X \mathcal R(x_1;|z|)\, |z|(\dd x_1) \leq \int_{\mathbb X} \e^{a(x_1) + \beta B(x_1)}\, |z|(\dd x_1) <\infty. 
\ee
which is the sense in which the power series for $\int_\mathbb X\rho(\dd x_1;z)$ converges.

Finally,  define $\tilde\nu(\dd q):= \mathcal R(q;|z|)\, |z|(\dd q)$ which by \eqref{dis2} is bounded by $\tilde \nu \leq \e^{a+\beta B} |z|$.

Now $\e^{a+\beta B}|z|$ is in $\mathsf V_b$ by the second condition in~\eqref{dissymmetry-condition:b} and therefore $\tilde \nu$  is in $\mathsf V_b$  as well. Thus we can bound 
\be \label{dissymmetry-converges}
\begin{aligned}
	&\sum_{n=2}^\infty \frac{1}{n!} \int_{\mathbb X^n} \Biggr(\sum_{m=2}^n m\sum_{\substack{L\subset[n]\\ \#L= m}} \bigl|D_{m}\bigl( (x_\ell)_{\ell\in L} \bigr)\bigr| \sum_{\substack{(J_\ell)_{\ell\in L}:\\ \dot \cup J_\ell = J}} \prod_{\ell \in L} \bigl|\varphi_{\#J_\ell +1}^\mathsf T\bigl( (x_j)_{j\in J_\ell\cup \{\ell\}}\bigr)\bigr|\Biggr)|z|^n(\dd x) \\
	&\quad = \sum_{m=2}^\infty \frac{1}{m!}\int_{\mathbb X^m} m \bigl| D_m(x_1,\ldots,x_m)\bigr| \Biggl( \prod_{i=1}^m \mathcal R(x_i;|z|)\Biggr) |z|^m(\dd \vect x)\\
	& \quad = \int_{\mathbb X} \Biggl( \sum_{m=1}^\infty \frac{1}{m!} \int_{\mathbb X^m} \bigl| D_{m+1}(q,x_1,\ldots,x_m) \bigr|\, \tilde \nu^m(\dd \vect x) \Biggr) \tilde \nu (\dd q) \\
	&\quad \leq \int_\mathbb X b(q) \tilde \nu(\dd q) \leq \int_{\mathbb X} b(q) \e^{a(q) + \beta B(q)} |z|(\dd q) <\infty,
\end{aligned} 
\ee
where in the third but last inequality we applied ~\eqref{virMb} with $\tilde \nu$ instead of $|\nu|$.
At the very end we have used again condition~\eqref{dissymmetry-condition:b}. 
This is the sense in which the third term converges. Note that the sense of convergence is strong enough, such that that also re-ordering of the terms is converging so long one does not break up $D_m$. As a consequence, Eq.~\eqref{eq:virmain2b} holds true not only as an equality of formal power series but also as an equality of convergent sums. 
\end{proof} 

\begin{proof}[Proof of Theorem~\ref{thm:GCFE}] 
	The standard line of reasoning is as follows: we check that the solution $z$ to the equation $\rho[z]=\nu$---which exists by Theorem~\ref{thm:virmain2}---is a maximizer in~\eqref{GCFELT}, deduce a formula for $\mathcal F_{\mathrm GC}[\nu]$ in terms of the maximizer $z$, plug in~\eqref{eq:virmain2a} and~\eqref{eq:virmain2b}, and obtain the statement. The full proof requires us to check that all steps are fully justified.
	
It is convenient to rewrite the definition~\eqref{GCFELT} as 
	\be \label{eq:gcfelt2}
		\beta \mathcal F_\mathrm{GC}[\nu]= \sup_{h} \Biggl( \int_\mathbb X h(x) \nu(\dd x) - \log \Xi[\e^h m] \Biggr),
	\ee
	where the supremum is taken over all measurable $h:\mathbb X\to \R\cup \{-\infty\}$ such that $\int_{\mathbb X} |h|\dd \nu<\infty$. 
	
	Let $\nu\in \mathsf V_b$ satisfy the assumptions of the theorem. By Theorem~\ref{thm:virmain2}, the measure $z_0:=\zeta[\nu]$ satisfies $\rho[z_0]=\nu$. Therefore, it is of the form $z_0 (\dd q)= \e^{h_0(q)} m(\dd q)$ with 
	$$
		h_0(q)= \log \frac{\dd \nu}{\dd m}(q) - \sum_{n=1}^\infty \frac{1}{n!}\int_{\mathbb X^n} D_{n+1}(q,x_1,\ldots,x_n) \nu^n(\dd q). 
	$$
	We check that $h_0$ is a maximizer in~\eqref{eq:gcfelt2}. 	As a preliminary observation, we note that $|h_0(q)|\leq |\log \frac{\dd \nu}{\dd m}(q)| + b(q)$ using \eqref{virMb}, therefore condition~\eqref{eq:nuassumption} yields $\int_\mathbb X |h_0|\dd \nu <\infty$. Thus $h_0$ does indeed belong to the set over which the supremum in~\eqref{eq:gcfelt2} is taken. 
	
	Let $h:\mathbb X\to\R\cup \{-\infty\}$ be another function with $\int_\mathbb X|h|\dd \nu<\infty$. We need to check that 
	\be \label{eq:gcfelt-max}
		\int_\mathbb X h(x) \nu(\dd x) - \log \Xi[\e^h m] 
			\leq \int_\mathbb X h_0(x) \nu(\dd x) - \log \Xi[\e^{h_0} m].
	\ee
	By the last condition in~\eqref{eq:nuassumption}, the measure $z_0 = \e^{h_0} m$ satisfies condition~\eqref{eq:fivo} and so $\Xi[\e^{h_0}m]<\infty$ and the right-hand side in~\eqref{eq:gcfelt-max} is finite. If $\Xi[\e^h m] =\infty$, then the inequality~\eqref{eq:gcfelt-max} holds trivially true. If $\Xi[\e^h m]<\infty$, then the inequality~\eqref{eq:gcfelt-max} is equivalent to 
	\be \label{eq:gcfelt-max2}
		\log \Xi[\e^{h} m]\geq \log \Xi[\e^{h_0} m] + \int_{\mathbb X} (h-h_0)\dd \nu
	\ee	
	and it will be checked with the help of convexity. Set
	$$
		g(t) := \log \Xi\bigl[\e^{(1-t) h_0+ t h} m\bigr], \quad t\in [0,1].
	$$
		It is a well-known consequence of H{\"o}lder's inequality that $g(t)$ is convex. 
		
		Next we check that the right derivative of $g$ at zero exists and is given by $g'(0) = \int_\mathbb X(h-h_0)\dd \nu$. We look at the derivative of $\exp(g(t))$ first. Set $h_t:= (1-t) h_0 + th$. We have 
	\be \label{eq:gcdiff}
		\Xi[\e^{h_t}m]  - \Xi[\e^{h_0} m] 
			 = \sum_{n=1}^\infty \frac{1}{n!} \int_{\mathbb X^n}  \bigl( \e^{\sum_{i=1}^n h_t(x_i)} - \e^{\sum_{i=1}^n h_0(x_i))}\bigr) \e^{- \beta H_n(x_1,\ldots,x_n)} m^n(\dd \vect x).
	\ee
	To facilitate differentiation, we check that configurations with infinite $h_t(x_i)$'s do not contribute. As $\int_{\mathbb X} e^h \dd m \leq \Xi[\e^{ h} m] < \infty$ we have $m$-a.e. that $h <\infty $. Furthermore, we can see that
	\begin{align}
		&\Biggl| \sum_{n=1}^\infty \frac{1}{n!} \int_{\mathbb X^n}  \bigl( \e^{\sum_{i=1}^n h_t(x_i)} - \e^{\sum_{i=1}^n h_0(x_i))}\bigr) \1_{\{\exists i:\, h(x_i) =-\infty\}} \e^{- \beta H_n(x_1,\ldots,x_n)} m^n(\dd \vect x)\Biggr| \notag \\
			&\quad =  \sum_{n=1}^\infty \frac{1}{n!} \int_{\mathbb X^n}  \1_{\{\exists i:\, h(x_i) =-\infty\}} \e^{- \beta H_n(x_1,\ldots,x_n)} z_0^n(\dd \vect x) \notag  \\
			& \quad \leq \int_\mathbb X \1_{\{ h(q) = - \infty\}} \rho(\dd q;z_0) = \int_{\mathbb X }\1_{\{ h(q) = -  \infty\}} \nu(\dd q), \label{eq:gcdiff2}
	\end{align}
where we used in the first equality that $\e^{\sum_{i=1}^n h_t(x_i)} \1_{\{\exists i:\, h(x_i) =-\infty\}} =0$ and in the inequality that $\nu = \rho[z_0]$.	By choice of $h=h_1$, the integral $\int_{\mathbb X} |h|\dd \nu$ is finite, hence $h>-\infty$, $\nu$-almost everywhere.  It follows that the last expression in~\eqref{eq:gcdiff2} vanishes, hence also all preceding expressions in the chain of inequalities vanish. Therefore we have that $m$-a.e. holds $|h(x)| < \infty $. The same holds for $h_0$. 
As $|h_t(x)|= \infty$ only if either  $|h_0(x)|$ or $|h_1(x)| = |h(x_i)|$ are infinite we have that 
\bes
		C:= \{x\in \mathbb X\mid |h(x)| < \infty \mbox{ or } |h_0(x)| < \infty \} 
	\ees
has full $m$-measure and hence $h_t$ is well-defined on $C$.
	The considerations above yield 
	\be\label{eq:gcdiff3}
		\Xi[\e^{h_t}m]  - \Xi[\e^{h_0} m] 
			 = \sum_{n=1}^\infty \frac{1}{n!} \int_{C^n}  \bigl( \e^{ \sum_{i=1}^n h_t (x_i)} -\e^{ \sum_{i=1}^n h_0 (x_i)} \bigr)  \e^{- \beta H_n(x_1,\ldots,x_n)} m^n(\dd \vect x)
	\ee
	for all $t\in [0,1]$.
	We also have as $\nu = \rho[z_0]$ that
	\be \label{eq:gcdiff4}
		\Xi[\e^{h_0} m] \int_\mathbb X (h_1- h_0) \dd \nu = \sum_{n=1}^\infty \frac{1}{n!} \int_{C^n}  \sum_{i=1}^n\bigl( h_1(x_i) -h_0(x_i) \bigr) \e^{\sum_{i=1}^n h_0(x_i)} \e^{- \beta H_n(x_1,\ldots,x_n)} m^n(\dd \vect x).
	\ee
	and therefore it holds that
	\begin{multline}  \label{eq:gcdiff5}
			\frac{1}{t} \left(\Xi[\e^{h_t}m]  - \Xi[\e^{h_0} m] \right) - \Xi[\e^{h_0} m]\int_\mathbb X (h_1 - h_0) \dd \nu \\
			= \sum_{n=1}^\infty \frac{1}{n!} \int_{C^n} 
			\frac1t\Biggl( \e^{\sum_{i=1}^n h_t(x_i)} - \bigl(1+   t\sum_{i=1}^n[h_1(x_i) -h_0(x_i)]\bigr)  \e^{\sum_{i=1}^n h_0(x_i)} \Biggr) \e^{- \beta  H_n(x_1,\ldots,x_n)} m^n(\dd \vect x).
	\end{multline} 
	Each integrand goes to zero as $t\to 0$, we need a $t$-independent integrable upper bound in order to apply dominated convergence. For $a,u\in\R$ and $t>0$ we have 
	$$
		\frac1t\Bigl|\e^{a+tu} - \e^a (1 + t u) \Bigr|=\frac1 t \e^a \Bigl|\int_0^{tu} \bigl( \e^s - 1\bigr) \dd s \Bigr| \leq 
			|u| \e^a \left\{ \begin{array}{ll} 
			\e^{tu} & \mbox{for } u >0 \\ 1 & \mbox{for } u \leq 0 \end{array}\right.  .
	$$
 If $u>0$, pick $\eps\in (0,1)$ and assume $t\in (0,1-\eps)$.   We apply the inequality $x\e^{-x} \leq \e^{-1}$ to $x =\eps u$ and find that the upper bound is $u \exp(a+tu) \leq (\eps \e)^{-1}\exp( a+(t+\eps)u) \leq (\eps \e)^{-1} u \exp (a+u)$. Altogether we find 
	$$
		\frac1t\Bigl|\e^{a+tu} - \e^a (1 + t u) \Bigr| 
		\leq |u|\e^a+ \frac{1}{\eps \e} \e^{a+u}.
	$$
	This inequality applied to $\eps =1/2$,  $a=\sum_i h_0(x_i)$ and $u = \sum_i (h_1(x_i) - h_0(x_i))$ yields, for $t\in(0,1/2)$, that the integrand in~\eqref{eq:gcdiff5} is bounded in absolute value by 
	$$
		\sum_{i=1}^n \bigl|h_1(x_i) - h_0(x_i)\bigr| \e^{\sum_{i=1}^n h_0(x_i)} +  \frac{1}{\eps \e}\, \e^{\sum_{i=1}^n h_1(x_i)}  \e^{- \beta  H_n(x_1,\ldots,x_n)}.
	$$
When one  integrates over $x_1,\ldots,x_n$, multiply with $\frac{1}{n!}$, sum over $n$, one obtains
	$$
		\int_C |h- h_0| \dd \nu + \Xi[\e^h m] <\infty.
	$$
	Thus we may apply dominated convergence to~\eqref{eq:gcdiff5} and find that indeed 
	$$
		\lim_{t\searrow 0} \frac{1}{t} 	\Bigl(\Xi[\e^{h_t}m]  - \Xi[\e^{h_0} m] \Bigr) = \Xi[\e^{h_0} m] \int_\mathbb X (h_1- h_0) \dd \nu 
	$$
	from which we deduce $g'(0) = \int_\mathbb X (h_1- h_0) \dd \nu$. We have already observed that $g(t)$ is convex and hence one has that $g(t) \geq g(0) + g'(0) t$, which for $t=1$ is precisely the inequality~\eqref{eq:gcfelt-max2}. It follows that $h_0$ is a maximizer in~\eqref{eq:gcfelt-max} and 
	\begin{multline} \label{gcfe-prefinal}
		\mathcal F_\mathrm{GC}[\nu] = \int_\mathbb X h_0 \dd \nu - \log \Xi[\zeta[\nu]]  \\
			= \int_\mathbb X \Bigl( \log \frac{\dd \nu}{\dd m} (q) - \sum_{n=1}^\infty \frac{1}{n!}\int_{\mathbb X^n} D_{n+1}(q,x_1,\ldots,x_n) \nu^n(\dd x)\Bigr) \, \dd \nu (q) - \log \Xi[\zeta[\nu]]. 
	\end{multline}
	The final step is to insert the expression for $\log \Xi[\zeta[\nu]]$ from Eq.~\eqref{eq:virmain2b} in Theorem~\ref{thm:virmain2}, keeping in mind that $\rho[\zeta[\nu]] = \nu$. This then yields \eqref{eq:free-energy}. 
	
	To justify the application of~\eqref{eq:virmain2b}, we could in principle impose conditions on $\nu$ that guarantee that $z_0=\zeta[\nu]$ satisfies the condition~\eqref{dissymmetry-condition:b} from Theorem~\ref{thm:virmain2}, however this would result in more restrictive conditions and therefore we take a slightly different approach. 
	We start from the formal power series identity 
	\be \label{dissymmetry2}
		\log\Xi\bigl( \zeta[\nu]\bigr) 
			=\int_\mathbb X \nu(\dd x_1) 
			- \sum_{n=2}^\infty \frac{1}{n!}\int_{\mathbb X^n} (n-1) D_n(x_1,\ldots,x_n) \prod_{i=1}^n \nu^n(\dd \vect x) 
	\ee
	which follows from~\eqref{eq:virmain2b} and $\rho[\zeta[\nu]] = \nu$. It is justified, \emph{as a formal power series identity}, without any conditions on $\nu$. Additional arguments are needed to ensure that~\eqref{dissymmetry2} holds true as an equality of convergent expressions. The exponential of the left-hand side of~\eqref{dissymmetry2} is the formal power series 
	\be \label{xinuzeta0}
		1 + \sum_{n=1}^\infty \frac{1}{n!} \int_{\mathbb X^n}\sum_{\substack{L\subset[n]\\ L\neq \varnothing}} \e^{-\beta H_{\#L}((x_\ell)_{\ell\in L})}\sum_{\substack{(J_\ell)_{\ell \in L}\\ \dot \bigcup J_\ell = [n]\setminus L}} \prod_{\ell \in L} t_n\bigl(x_\ell;(x_j)_{j\in J_\ell}\bigr) \nu^n(\dd \vect x)
	\ee
	see Eq.~\eqref{fcomp2} in Appendix~\eqref{ApFormal}. The set $L$ is non-empty but $J_\ell = \varnothing$ is allowed (we agree $t_0=1$).  We have 
	\begin{multline} \label{xinuzeta}
		1+ \sum_{n=1}^\infty \frac{1}{n!} \int_{\mathbb X^n} \sum_{\substack{L\subset[n]\\ L\neq \varnothing}} \e^{-\beta H_{\#L}((x_\ell)_{\ell\in L})}\sum_{\substack{(J_\ell)_{\ell \in L}:\\ \dot \bigcup J_\ell = [n]\setminus L}} \prod_{\ell \in L} \bigl|t_{\#J_\ell} \bigl(x_\ell;(x_j)_{j\in J_\ell}\bigr)\bigr|\, |\nu|^n(\dd \vect x) \\
		\quad = 1+ \sum_{n=1}^\infty \frac{1}{n!} \int_{\mathbb X^n} \e^{-\beta H_n(x_1,\ldots,x_n)} \prod_{i=1}^n \mu(\dd \vect x) 
	\end{multline}
	with 
	$$
		\mu(\dd q):= |\nu|(\dd q) \Biggl( 1+ \sum_{n=1}^\infty \frac{1}{n!} \int_{\mathbb X^n} \bigl|t_n\bigl(q;x_1,\ldots,x_n\bigr)\bigr|\, |\nu|^n(\dd \vect x)\Biggr).
	$$
	The term in parentheses is smaller or equal to $\exp(b(q))$ by our assumption $\nu \in \mathsf V_b$ and Lemma~\ref{lem:abkey}, therefore 
	$$
		\int_\mathbb X \e^{B(q)}\mu(\dd q)\leq \int_\mathbb X \e^{B(q)+b(q)} |\nu|(\dd q)<\infty
	$$	
	by the last assumption on $\nu$ in~\eqref{eq:nuassumption}. 
	It follows that $\mu$ satisfies the finite-volume condition~\eqref{eq:fivo}, hence $\Xi(\mu)$ is finite and thus~\eqref{xinuzeta} is finite. It follows that~\eqref{xinuzeta0} is equal to $\Xi[\zeta[\nu]]$ not just as a formal power series but as an equality of convergent series.
	
	Similar considerations apply to the right-hand side of~\eqref{dissymmetry2}. It follows that~\eqref{dissymmetry2} holds true as an equality of convergent series. We plug the expression for $\Xi[\zeta[\nu]]$ from~\eqref{dissymmetry2} into the formula~\eqref{gcfe-prefinal} and obtain the expression~\eqref{eq:free-energy} for the free energy. 
\end{proof} 

\section{Examples} \label{sec:examples}

\subsection{Homogeneous gas} \label{sec:hom} 

Consider a homogeneous gas of particles in a domain $\Lambda\subset \R^d$, interacting via a  translationally invariant pair potential $V(x,y) = v(x-y)$, with $v(x) = v(-x)$. The potential is assumed to be stable, 
$$
	\sum_{1\leq i < j \leq N} v(x_i-x_j) \geq - B N
$$
for some $B\geq 0$, all $N\geq 2$, and all $x_1,\ldots, x_N\in \R^d$. Furthermore, we assume
$$
	\bar C(\beta):= \int_{\R^d} \bigl(1 - \e^{- \beta |v(x)|} \bigr) \dd x <\infty.
$$
Further assume that $\inf v\geq - B^*$ for some $B^*\in (0,\infty)$. 
Mayer's irreducible cluster integrals are defined as 
$$
	\beta_{n}:= \frac{1}{n!} \int_{(\R^d)^n} \sum_{g\in \mathcal B_{n+1}} \prod_{\{i,j\}\in E(g)} \bigl( \e^{-\beta v(x_i-x_j)} - 1\bigr) \dd x_2\cdots \dd x_{n+1},\quad x_1:=0, 
$$
which in terms of the coefficients $D_n$ from~\eqref{eq:Dn}, can be expressed as
\be \label{eq:betanDn}
	\beta_n = \frac{1}{n!}\int_{\mathbb X^n} D_{n+1}(0,x_1,\ldots, x_n) \dd x_1\cdots \dd x_n. 
\ee
The grand-canonical partition function $\Xi_\Lambda(\beta,z)$ at inverse temperature $\beta>0$ and activity $z>0$ is defined in the usual way, and the pressure is given by 
\be \label{hom-thermolim}
	\beta p_\beta(z):= \lim_{\Lambda\nearrow \R^d} \frac{1}{|\Lambda|} \log \Xi_\Lambda(\beta,z),
\ee
with the limit taken along van Hove sequences~\cite{ruelle1969book}. Further set 
\be \label{hom-density}
	\rho_\beta(z):= z\frac{\partial}{\partial z} \beta p_\beta (z). 
\ee
It is well-known~\cite{ruelle1969book} that if  $C(\beta)\e^{2\beta B}  |z|\leq \frac{1}{\e}$, then the limit~\eqref{hom-thermolim} and the derivative~\eqref{hom-density} exist, moreover they define functions that are analytic in $C(\beta) \e^{2\beta B}|z|<\frac{1}{\e}$ (at least), we use the same letters for the analytic extensions to the complex disk. 
We fix $\beta>0$ and drop the $\beta$-dependence from the notation in $p_\beta(z)$ and $\rho_\beta(z)$. 

\begin{theorem} \label{thm:hom-virmain}\hfill
	\begin{enumerate}
		\item [(a)] If $\nu \in \C$ satisfies $\bar C(\beta)\e^{\beta [B+B^*]} |\nu|\leq \frac{1}{2\e}$, then $\sum_{n=1}^\infty |\beta_n \nu^n|\leq \frac12$. In particular, the radius of convergence $R_{\mathrm{vir}}$ of $\sum_n \beta_n \nu^n$  is bounded from below by 
	\be \label{hom-virmain}
		R_{\mathrm{vir}} \geq R^*:=\frac{1}{2 \e} \frac{1}{\e^{\beta (B+B^*)} \bar C(\beta)}.
	\ee
		\item [(b)] There exists some neighborhood $\mathcal O$ of the origin with
$$
 \Bigl \{z\in \C\, \Big|\, \bar C(\beta) \e^{\beta (B+B^*)}|z| < \frac{1}{\e \e^{2/\e}} \Bigr\}  \subset \mathcal O\subset \Bigl\{z\in \C\Bigl| \bar C(\beta) \e^{\beta (B+B^*)}\,\, |z| < \frac{1}{2\sqrt{e}} \Bigr\} 
$$
		such that $\rho(\cdot)$ is a bijection from $\mathcal O$ onto the open ball $B(0,	R^*)$, with inverse 
		$$
			z(\rho) = \rho \exp\Bigl( - \sum_{n=1}^\infty \beta_n \rho^n\Bigr).
		$$ 
		\item [(c)] For all $z\in \mathcal O$, we have 
	\be \label{hom-virmain2b}
		\beta p(z) = \rho(z) + \sum_{n=1}^\infty\frac{n\beta_n}{n+1} \rho(z)^{n+1}. 
	\ee
		\item [(d)] For all $\rho\in (0,R^*)$, the Helmholtz free energy $f (\rho):= \sup_{z>0} (\beta^{-1}\rho\log z- p(z))$ is given by 
		\be \label{eq:frho}
			\beta f(\rho) =  \rho(\log \rho - 1) - \sum_{n=1}^\infty \frac{\beta_{n}}{n+1}\rho^{n+1}. 
		\ee
	\end{enumerate} 
\end{theorem} 

\noindent Let us compare our result to what was known before. The bound~\eqref{hom-virmain} should be contrasted with the best known bound
\be \label{virbest} 
	R_\mathrm{vir} \geq R_0:=\frac{k}{\bar C(\beta) \exp(\beta \bar B)}
\ee
where 
\be \label{virconstant}
	k:= \max_{0\leq w\leq 1} (2 \e^{-w} - 1)w
	\geq 0.14476
\ee
(the lower bound is sharp, it is actually $k = \frac{(W(\e/2) - 1)^2}{W(\e/2)}$, cf. \cite{tate2013virial}) and 
\be \label{basuev}
	\bar B:= \inf_{n\geq 2}\frac{1}{n-1} \inf_{x_1,\ldots,x_n\in \R^d} H_n(x_1,\ldots,x_n).
\ee
For non-negative pair potentials, we have $\bar B=0$ and~\eqref{virbest} coincides with the lower bound proven by Lebowitz and Penrose~\cite{lebowitz-penrose1964virial}, who also proved the lower bound in~\eqref{virconstant}. 
For attractive pair potentials, the bound~\eqref{virbest} is an improvement on the bound from~\cite{lebowitz-penrose1964virial}, which was  proven in~\cite{procacci2017correction}, where the constant $\bar B$ is called the \emph{Basuev stability constant}. The constant $\bar B$ also enters an asymptotic \emph{upper} bound to $R_\mathrm{vir}$ as $\beta\to \infty$, see~\cite[Theorem 2.8]{jansen2012mayer}. 

Let us compare our bound~\eqref{hom-virmain} with~\eqref{virbest}. It differs in two ways:
it has a different constant $\frac{1}{2\e}$ and a different exponential $\exp( - \beta (B^*+B))$. Our constant $\frac{1}{2\e}$ is better but for attractive interactions  our exponential  in general is worse. As a consequence, for non-negative interactions, our bound yields a considerable improvement over the bound from~\cite{lebowitz-penrose1964virial} and hence~\eqref{virbest}
$$
	0.1840 > \frac{1}{2\e} > 0.1839 > 0.14477 > k > 0.14476,
$$
therefore our bound improves substantially all known bound. The improvement subsists for attractive interactions with small $\beta$. For large $\beta$ or strong interactions, the bound~\eqref{virbest} due to~\cite{procacci2017correction} trumps ours.\\

\begin{remark}[Attractive potentials]
	Additional work is needed to see whether our exponent $\exp(- \beta(B+ B^*))$ in~\eqref{hom-virmain} can be replaced by the exponent $\exp( - \beta \bar B)$ as in~\eqref{virbest}. This is related to the fact that bounding $b_n$'s in the Mayer expansion $\rho(z) =\sum_{n=1}^\infty n b_n z^n$ may sometimes be better than bounding $a_n$ in the representation $\rho(z) = z\exp( - \sum_{n=1}^\infty a_n z^n)$. Indeed, in our approach, the factor $\exp( -\beta B^*)$ comes up in Lemma~\ref{lem:abkey} where, in order to  write $\rho(z)/z$ the density as an exponential $\exp(-A(z))$ and bound the exponent, we split the expansion of $A$ and we get an additional factor $\exp(\beta B^*)$ in Eq.~\eqref{usualmagicformula}. 
\end{remark} 

\begin{remark} [Relation with Lagrange inversion]
	After the proof of Theorem~\ref{thm:hom-virmain} we will explain how to recover our bound~\eqref{hom-virmain} in the case $B=0$ based on a slightly different treatment of the Lagrange inversion from~\cite{lebowitz-penrose1964virial}, and where exactly our gain is achieved. 
\end{remark} 

\begin{remark}[Further improvements for non-negative pair potentials]
	The factor $\frac{1}{2\e}$ could be further improved using our techniques combining them with the  refined tree-graph inequality from~\cite{fernandez-procacci-scoppola2007}, i.e., working with  trees where children communicate, resulting in additional constraints on trees.  Instead of the generating function of the trees, one has to consider the solution of the equation
	\[
	G(s) = s \left( 1+  \sum_{k=0}^\infty \frac{1}{k!}\left(  G(s) \right)^k \tilde{g}_d(k) \right),
	\]
where $\tilde{g}_d(k)$ are defined as in \cite{fernandez-procacci-scoppola2007}, where it was used for the activity expansion.	
	Doing so, for hard disks, that is $d=2$, one obtains $0.196$ as  a lower bound for the radius of convergence of the virial expansion instead of $\frac{1}{2\e}\approx 0.184$.
\end{remark}

\begin{proof} [Proof of Theorem~\ref{thm:hom-virmain}]
	We apply the considerations from Section~\ref{sec:virial} to the case $\mathbb X= \R^d$, $\mathcal X$ the Borel sets, and specialize to translationally invariant measures $z(\dd x) = z \dd x$ with a constant scalar $z$. For such a measure the measure $\rho(\dd q;z)$ given by $\exp( - A(q;z)) z(\dd q)$ is translationally invariant as well, we write $\rho(\dd q; z) = \rho(z) \dd q$ and note that $\rho(z)$ is equal to the limit~\eqref{hom-density}, moreover 
	$\rho(z) = z \exp( - A(z))$ with 
	$$
		A(z) = -  \sum_{n=1}^\infty \frac{z^n}{n!}\int_{(\R^d)^n}  \Biggl[ \prod_{i=1}^n (1+ f(0,x_i)) - 1\Biggr] \varphi_n^\mathsf T(x_1,\ldots, x_n) \dd x_1\cdots \dd x_n.
	$$
	Conversely, if $\nu(\dd q) = \nu \dd q$ is a translationally invariant measure, then the inverse $\zeta(\dd q;\nu)$ from is translationally invariant as well. 
	
	By Theorem~\ref{thm:virmain2} applied to $\nu(\dd x) = \nu \dd x$, constant functions $a$ and $b$,  if the number	$\nu\in \C$ satisfies 
	\be \label{hom-suff}
		\bar C(\beta)\e^{\beta (B+B^*)} \e^{a+b} |\nu| \leq a
	\ee
	for some $a,b\geq 0$ with $a\leq b$, then  $\nu \in \mathsf V_b$
	\be \label{eq:hom-virmain1}
		\sum_{n=1}^\infty |\beta_n|\, |\nu|^n  \leq b
	\ee
	(remember~\eqref{eq:betanDn} and~\eqref{virMb}). Condition~\eqref{hom-suff} is further evaluated as 
	$$
		\bar C(\beta)\e^{\beta (B+B^*)}|\nu|\leq \sup\{ a\, \e^{-a - b}\mid b\geq a \geq 0\} = \sup\{ a\, \e^{-2 a}\mid a\geq 0\} = \frac{1}{2\e}.
	$$
	Therefore if $\bar C(\beta)\e^{\beta (B+B^*)}|\nu|\leq \frac{1}{2\e}$, then condition~\eqref{hom-suff} holds true with $a= b=\frac 12$ and Eq.~\eqref{eq:hom-virmain1} holds true with $b=\frac12$. 
Part (a) of the theorem follows. 

Part (b) follows from the first part of Theorem~\ref{thm:virmain2}, with $b=1/2$ and $\mathsf V_{1/2} = B(0,R^*)$. Then $|\zeta[\nu]| = |\nu| \ \left| T^\circ(\nu) \right| \leq R^* e^b $ by Theorem~\ref{thm:main1}. For $z\in \mathbb C$ there exists $a,b$ with $0 \leq  a\leq b$ such that $e^{a}|z| \in \mathsf V_b$ if and only if
$ C(\beta)\e^{\beta (B+B^*)} |z|\leq \frac{1}{\e\e^{2/\e}} $. Note that $e^{a} |z|\in \mathsf V_b$ means that \eqref{suffsuff} holds for an $0 \leq \tilde{a} \leq b$ instead of $a$ and not necessarily $a = \tilde{a}$.
	
	For part (c), we note that the validity of~\eqref{hom-virmain2b} for sufficiently small $|z|$ is already known~\cite{lebowitz-penrose1964virial}. Alternatively, we may deduce from Theorem~\ref{thm:virmain2} by working first in finite volume and then taking the infinite-volume limit. This way of proceding guarantees the validity of~\eqref{hom-virmain2b} under the additional condition $\e^{a+ \beta B}|z|<R^*$ for some $0\leq a\leq \frac12$. The additional condition is eliminated by invoking analyticity: The left and right sides of~\eqref{hom-virmain2b} define functions of $z$ that are analytic in $\mathcal O$ and coincide on some non-empty open ball, therefore they are equal on all of $\mathcal O$. 
	
	Theorem~\ref{thm:GCFE}, using that in part (a) we have shown  $\rho \in \mathsf V_b$ whenever $0 \leq |\rho| \leq R^*$, gives part (d) of the theorem  in finite volume.  
		By part (a) the radius of convergence of the series in~\eqref{eq:frho} is independent of the volume. 	Combining with the translation invariance we see that the right-hand side of \eqref{eq:free-energy}, divided by the volume, converges to the right-hand side of~\eqref{eq:frho} in the infinite-volume limit. 	Furthermore, that the free-energy is the Legendre transform of the pressure can be extended to the infinite volume as well. 	
	 We note that the validity of~\eqref{hom-virmain2b} for sufficiently small $|z|$ was already known~\cite{lebowitz-penrose1964virial}.
\end{proof} 

\noindent Let us provide an alternative derivation of the bound~\eqref{hom-virmain} for non-negative potentials ($B=0$). The key point in~\cite{lebowitz-penrose1964virial} is a lower bound for the radius of convergence $R_\mathrm{vir}$ of the expansion in $\rho$ as 
\be \label{LP1}
	 R_\mathrm{vir}  \geq \sup_{R \geq >r\geq 0}\inf_{|z|=r}|\rho(z)|
\ee
which is derived in \cite{lebowitz-penrose1964virial} using a Lagrange inversion /tk{, where $R$ is the convergence radius of $\rho$ at zero}. A lower bound for $R_\mathrm{vir}$ is then deduced from a lower bound for $|\rho(z)|$. This is done in~\cite{lebowitz-penrose1964virial} (and also in~\cite{tate2013virial})  with the help of the triangle inequality $|\rho(z)|\geq |z| - |\rho(z) - z|$.
It turns out that if, instead, one uses the exponential structure 
$\rho(z)= z \e^{-A(z)}$	
and an upper bound for $|A(z)|$  one can recover our bound~\eqref{hom-virmain} from~\eqref{LP1}. Let us explain the strategy. Our aim is to prove the following chain of inequalities
\begin{align}\label{LP}
R_\mathrm{vir} & \geq \sup_{R \geq r \geq 0} \inf_{|z|=r} 
\left| z\e^{-A(z)}\right|
 \geq \sup_{R \geq r\geq 0} r \e^{-T(\bar C(\beta) r)}=\frac{1}{2\e}\frac{1}{\bar C(\beta)}
\end{align}
where $T(z) =\sum_{n\geq 1} \frac{n^{n-1}}{n!} z^n$ is the generating function of labelled rooted trees (equivalently, $T(z) = - W(-z)$ with $W$ the Lambert function) and $R$ is the radius of convergence of the analytic function $A(z)$. 

The first inequality in Eq.~\eqref{LP}, merely uses the idea $\rho(z)= z \e^{-A(z)}$.
The second inequality can be derived in several ways. It follows directly from Penrose's tree-graph inequality, namely  as in Lemma~\ref{lem:abkey} use estimate~\eqref{usualmagicformula} with $B^*=0$, then by the tree-graph inequality $\mathcal{R}(q;z) \leq \frac{T(\bar C(\beta) |z|)}{\bar C(\beta) |z|}$. Hence one gets $|A(z)| \leq T(\bar C(\beta) |z|)$. Alternatively, following more the type of results used in this article, one gets from the inductive proof of Theorem 2.1 in \cite{poghosyan-ueltschi2009}, cf. Lemma~\ref{lem:abkey} as well,  that
$|A(z)|\leq a$ whenever $\bar C(\beta) \e^a |z|\leq a$. It remains to optimize over $a$.Recall that $T(s)$ is converges on $[0,1/\e]$ with $T(1/\e)=1$ and is also the branch of the real solution of the relation $T(s) = s \e^{T(s)}$ with $T(s) \leq 1$. Hence, $T(s)$ satisfies, for  $s\geq 0$, 
\be \label{trick0}
	T(s) = \inf \{ a\mid a\in [0,1],\ a \e^{-a}\geq s \}
\ee 
Since $T(s)$ diverges for $s>1/\e$, Eq.~\eqref{trick0} stays true 
 for $s> 1/\e$ if we interpret the infimum of the empty set as infinity. Equation~\eqref{trick0} follows from the relation $T(s) = s \e^{T(s)}$ solved by $T$, the bound $T(s) \leq T(1/e) =1$  and the the fact that $a\mapsto a \e^{-a}$ is strictly increasing on $[0,1]$. 
 Consequently, using ~\eqref{trick0} we get
\be\label{trick}
|A(z)|\leq \inf\bigl\{a \, \big|\, \bar C(\beta) \e^a |z|\leq a\leq 1\bigr\}=T\bigl(\bar C(\beta) |z|\bigr).
\ee
Hence the bound is finite and thus the radius of convergence of $A(z)$ is $R < \frac{1}{\e \bar C(\beta) }$.

Finally, the third inequality can be derived using again~\eqref{trick0}  and $T(s) = s\e^{T(s)}$, we have
\begin{align}
\sup_{1/\e \geq s\geq 0} s\, \e^{-T(s)} & = \sup_{s\geq 0} s\e^{-T(s)} = \sup_{s\geq 0} \frac{s^2}{T(s)} = \sup\left\{ \left.\frac{s^2}{a}\right| s\geq 0,\ a\in[0,1],\ s\leq a\e^{-a} \right\} \notag \\
	&= \sup_{a \in[0,1]} \frac{(a\e^{-a})^2}{a}=\frac{1}{2\e}.
\end{align}
Setting $s=\bar C(\beta) r$ we deduce the final bound in \eqref{LP}, which is the same as~\eqref{hom-virmain} in the case of non-negative potential. 

\subsection{Inhomogeneous gas} 
Here we start from a homogeneous gas with fixed reference activity $z_0>0$ and then add an external potential $V_{\mathrm{ext}}(x)$. 
The grand-canonical partition function in some bounded domain $\Lambda$ becomes 
\be \label{inhom-partitionfct}
	\Xi_\Lambda = \Xi_\Lambda(\beta, z_0,V_\mathrm{ext}) = 1+ \sum_{n=1}^\infty\frac{z_0^n}{n!}\int_{\Lambda^n} \e^{- \beta [\sum_{1\leq i<j \leq n} v(x_i -x_j) + \sum_{i=1}^n V_\mathrm{ext}(x_i)]} \dd x_1\cdots \dd x_n
\ee
and the density is given by 
\be \label{inhom-density}
	\rho_\Lambda(x_0; V_\mathrm{ext}):= z_0 \e^{-\beta V_\mathrm{ext}(x_0)}\, \frac{1}{\Xi_\Lambda}\Bigl( 1+ \sum_{n=1}^\infty \frac{z_0^n}{n!}\int_{\Lambda^n} \e^{- \beta [\sum_{0\leq i<j \leq n} v(x_i -x_j) + \sum_{i=1}^n V_\mathrm{ext}(x_i)]} \dd x_1\cdots \dd x_n\Bigr).
\ee
Eq.~\eqref{inhom-density} can be brought into the form from Section~\ref{sec:virial}: let 
\be \label{inhom-activity}
	z(x):= z_0 \exp\Bigl( - \beta V_\mathrm{ext}(x)\Bigr),
\ee
then 
\be \label{inhom-density2}
	\rho_\Lambda(x_0; V_\mathrm{ext}):= z(x_0) \frac{1}{\Xi_\Lambda}\Bigl( 1+ \sum_{n=1}^\infty \frac{1}{n!}\int_{\Lambda^n} \e^{- \beta \sum_{0\leq i<j \leq n} v(x_i -x_j)} \prod_{i=1}^n z(x_i) \dd x_1\cdots \dd x_n\Bigr),
\ee
similarly for the partition function. It follows from the tree-graph inequality in~\cite{procacci-yuhjtman2017} that if 
\be \label{inhom-PU}
	\int_{\R^d}\bar f(x,y)\, \e^{a(y)+ \beta B} z(y) \dd y = z_0	\int_{\R^d}{\bar f}(x,y)\, \e^{a(y)+ \beta B} \e^{- \beta V_\mathrm{ext}(y)} \dd y \leq a(x)
\ee
for some $a:\R^d\to \R_+$ and all $x \in \R^d$, then the limit 
$$
	\rho(x_0;V_\mathrm{ext}) = \lim_{\Lambda\nearrow \R^d} \rho_\Lambda(x_0;V_\mathrm{ext})
$$
exists and is given by the usual combinatorial formulas, with position-dependent activity $z(x)$ given in~\eqref{inhom-activity}. 

It is a classical problem to ask whether,  given a density profile $\rho(x)$, there exists a background potential $V_\mathrm{ext}(x)$ such that the density profile $\rho(x;V_\mathrm{ext})$ in the associated grand-canonical ensemble is equal to the given profile $\rho(x)$. In view of~\eqref{inhom-activity}, Theorem~\ref{thm:virmain2} has direct implications for this problem when activities converge. For results without cluster expansions, see~\cite{chayes-chayes-lieb1984}. 

\begin{theorem}  \label{thm:inhom}
	Fix $\beta,z_0>0$ and a pair potential $v(x-y)$ with stability constant $B$ and lower bound $\inf v \geq - B^*>-\infty$. Let $\rho:\Lambda\to \R_+$ be a measurable function such that 
	\be \label{inhom-suff}
		\int_{\R^d} \bar f(x,y)\, \e^{a(y)+ \beta (B+B^*) + b(y)} \rho(y) \dd y \leq a(x)
	\ee
	for all $x\in \R^d$ and some functions $a,b:\R^d\to \R_+$ with $a\leq b$ pointwise. Then there exists a unique (up to null sets) background potential $V_\mathrm{ext}:\Lambda\to \R\cup\{\infty\}$ that satisfies~\eqref{inhom-PU} and such that $\rho(q;V_\mathrm{ext}) = \rho(q)$ for Lebesgue-almost all $q$.  It is given by 
	\be \label{inhom-V}
		\beta V_\mathrm{ext}(q) = \log z_0 - \log \rho(q) + \sum_{n=1}^\infty\frac{1}{n!}\int_{\mathbb{R}^{dn}} D_{n+1}(q,x_1,\ldots,x_n) \rho(x_1)\cdots \rho(x_n) \dd x_1\cdots \dd x_n
	\ee
	with absolutely convergent integrals and sum.
\end{theorem} 

\noindent A sufficient condition for~\eqref{inhom-suff} to hold true is that 
$
	\bar C(\beta) \e^{\beta B} ||\rho||_\infty \leq \frac{1}{2\e}
$
(pick $a=b\equiv \frac12$). In fact one easily checks that, if we are interested in bounded density profiles only, we are in the situation where a direct application of the Banach inversion theorem (Theorem~\ref{thm:banach}) is possible. 

\begin{proof}
	The absolute convergence of the series in~\eqref{inhom-V} follows right away from Theorem~\ref{thm:virmain1} applied to $\nu(\dd x) = \rho(x) \dd x$. By Theorem~\ref{thm:virmain2}, there is a unique measure $z(\dd q)$ in the domain of convergence $\mathscr D(A)$ such that $\nu(\dd q) = \rho(\dd q;z)$, with $\rho(\dd q;z)$ the density at activity $z(\dd x)$ for the interaction potential $v(x-y)$. Moreover the activity is given by Eq.~\eqref{eq:virmain2a}, which after plugging in $\nu(\dd q) = \rho(q) \dd q$ becomes $z(\dd q) = z(q) \dd q$ with 
	\be \label{inhom-zq}
		z(q) = \rho(q)\exp\Biggl( - \sum_{n=1}^\infty\frac{1}{n!}\int_{\Lambda^n} D_{n+1}(q,x_1,\ldots,x_n) \rho(x_1)\cdots \rho(x_n) \dd x_1\cdots \dd x_n\Biggr).
	\ee
	We adopt~\eqref{inhom-activity} as a definition of the external potential, then $\beta V_\mathrm{ext}(q) = \log z_0 - \log z(q)$ and $V_\mathrm{ext}(q)$ is given by~\eqref{inhom-V}. It satisfies $\rho(q;V_\mathrm{ext}) = \rho(q)$ by the definition~\eqref{inhom-zq} of $z(q)$ and $V_\mathrm{ext}$. Condition~\eqref{inhom-PU} follows rom \eqref{virMb} as then $	|z(q)| \leq |\rho(q)| e^{b(q)}$ 
	and thus \eqref{inhom-suff} implies~\eqref{inhom-PU}.
\end{proof} 

\subsection{Mixture of hard spheres}  
Consider a mixture of hard spheres with radii $R_1,R_2,\ldots$, for example, $R_k = k^{1/d}$.  The activity $z_k$ of the sphere depends on the type $k$ but otherwise the system is homogeneous. To bring the model into the form from Section~\ref{sec:virial}, let $\mathbb X= \R^d\times \N$, with $(x,k)$ representing a sphere of radius $R_k$ centered at $x$. We consider measures $z$ informally given by $z= \oplus_{k\in\N} z_k \dd x$. More precisely, $\int_\mathbb X  h \dd z = \sum_{k=1}^\infty \int_{\R^d} h(x,k) z_k \dd x$ for every non-negative test function $h$. The interaction is hard core exclusion
$$
	V\bigl( (x,k),(y,\ell)\bigr) = \begin{cases}
		\infty, &\quad |x-y|\leq R_k+R_\ell,\\
		0, &\quad \text{else}. 
	\end{cases} 
$$
Let $p((z_k)_{k\in \N})$ be the infinite-volume pressure and $\rho_k((z_j)_{j\in \N}):= z_k \frac{\partial p}{\partial z_k}((z_j)_{j\in \N})$. A sufficient condition for the convergence of the activity expansion of the pressure is 
\be \label{spheres-activity}
	\sum_{\ell=1}^\infty|z_\ell|\, |B(0,R_k +R_\ell)|\, \e^{a_\ell} \leq a_k,
\ee
for some non-negative sequence $(a_j)_{j\in \N}$ of positive numbers and all $k\in \N$, as is easily checked from~\cite{ueltschi2004}.

\begin{theorem}
	Suppose that $(\rho_k)_{k\in\N}\in \C^\N$ satisfies 
	\be \label{spheres-density}
		\sum_{\ell=1}^\infty |\rho_\ell|\, |B(0, R_k+R_\ell)| \e^{a_\ell + b_		\ell} \leq a_k,
	\ee
	for all $k\in \N$ and two sequences $(a_j)$, $(b_j)$ with $b_j\geq a_j\geq0$ for all $j\in \N$. Then there exists a unique sequence $(z_k)_{k\in\N}$ with $\rho_j((z_k)_{k\in \N}) = \rho_j$ for all $j\in \N$ and such that condition~\eqref{spheres-activity} holds. It is given by 
	\be \label{spheres-inversion}
		z_k = \rho_k \exp\Biggl( - \sum_{n=1}^\infty \frac{1}{n!} \sum_{k_1,\ldots,k_n\in \N} \int_{(\R^d)^n} D_{n+1}\bigl((0,k),(x_1,k_1),\ldots,(x_n,k_n)\bigr) \rho_{k_1}\cdots \rho_{k_n} \dd \vect x \Biggr). 
	\ee
\end{theorem} 

\noindent The coefficients $D_n$ are given by sums over $2$-connected graphs as in~\eqref{eq:Dn}. The sum in the exponential in~\eqref{spheres-inversion}, with absolute values inside the integral, is bounded by $b_k$. 

 The theorem is deduced from Theorems~\ref{thm:virmain1} and~\ref{thm:virmain2},  the details are left to the reader. 

\subsection{Flexible molecules. Liquid crystals}
Finally we come to a system of objects with internal degrees of freedom: we assume that the space $\mathbb X$ is of the form $\mathbb X  = \Lambda \times S$ with $\Lambda\subset \R^d$ a bounded domain.\footnote{We could also allow for spaces $\mathbb X = \sqcup_{k\in \N}(\Lambda\times S_k)$ representing a multi-species system where each species $k$ has its own spin space $S_k$, but for simplicity we stick to the single-species case.} The space $S$ represents internal degrees of freedom (spin, orientation, shape of a molecule...). For example, we could take $S$ as the projective space $\mathbb P^{d-1}$ (i.e., $\R^d\setminus \{0\}$ with identification of parallel vectors) and think of $(x,\vec u)$ as a thin rod centered at $x$ with orientiation $\vec u$. Such a model is often used for the study of liquid crystals~\cite{onsager1949}. 

Suppose we are given a reference measure $m$ on $\mathbb X$ that is of the form $m(\dd (x,\sigma)) = \dd x \, \lambda(\dd \sigma)$, i.e., it is the product of the Lebesgue measure on $\Lambda$ and a reference measure $\lambda$ on $S$ (e.g. a uniform measure on orientations of thin rods). To simplify formulas, we write $\dd \sigma$ instead of $\lambda(\dd \sigma)$. The pair potential $V((x,\sigma), (y,\tau))$ is a function of both position and internal degree of freedom.  

Following Onsager, one could work in a multi-species canonical ensemble, where each species represents a discretized orientation. In such  a setup,  deriving the canonical free energy is immediate following~\cite{pulvirenti-tsagkaro2012}. 
One can easily derive a functional for continuous orientations, using our techniques presented here, that is, to start in the grand-canonical ensemble, and obtain the grand-canonical free energy via Legendre transform and inversion of the density-activity relation, which is precisely the definition~\eqref{GCFELT} for $\mathcal F_{\mathrm{GC}}[\nu]$. Let us write $\nu(\dd(x,\sigma)) = \rho(x,\sigma) \dd x\dd \sigma$ and, by a slight abuse of language, $\mathcal F_{\mathrm{GC}}[\rho]$ instead of $\mathcal F_{\mathrm{GC}}[\nu]$.

For simplicity we prove results for non-negative pair potentials $V$ only but note that our general theorems lead just as easily to stable pair potential. 

\begin{theorem} \label{thm:onsager}
	Let $V\geq 0 $ and $\rho:\mathbb X\to \R_+$. 
	Suppose there exist weight functions $a,b:\mathbb X\to \R_+$ with $b\geq a$.  Suppose that $\rho:\Lambda\times S\to \R_+$ satisfies 
	$$
		\int_{\Lambda\times S} \rho(y,\sigma)\, \Bigl(1- \e^{-\beta V((x,\sigma), (y,\tau)) }\Bigr) \e^{a(y,\sigma)+b(y,\sigma)} \dd y\, \dd \tau  \leq a(x,\sigma),
	$$
	for all $(x,\sigma)\in \Lambda\times S$, and 
	$$
		\int_{\Lambda\times S} 	\rho(x,\sigma) \Bigl( \bigl|\log \rho(x,\sigma) \bigr| + 1+ b(x,\sigma) + \e^{b(x,\sigma)} \Bigr) \dd x\, \dd \sigma <\infty. 
	$$	
	Then 
	\begin{multline*}
		\beta \mathcal F_\Lambda[\rho]
			 = \int_\Lambda \rho(x,\sigma) \bigl[\log \rho(x,\sigma) - 1\bigr]\dd x \dd \sigma \\
		 - \sum_{n=2}^\infty \frac{1}{n!} \int_{\Lambda^n}\int_{S^n} D_n\bigl((x_1,\sigma_1),\ldots, (x_n,\sigma_n) \bigr) \prod_{i=1}^n \rho(x_i,\sigma_i) \dd \vect x \dd \vect \sigma. 
	\end{multline*} 
	with absolutely convergent integrals and sum.
\end{theorem} 

\begin{proof} 
	The theorem is an immediate consequence of Theorem~\ref{thm:GCFE}. 
\end{proof} 

When we think of rods with an orientiation, we may specialize to situations where there is translational invariance but not necessarily rotational invariance: 

\begin{cor} \label{cor:feinternal}
	Assume that $\rho(x,\sigma) = \rho_0 p(\sigma)$ for some scalar $\rho_0>0$ and non-negative $p:S\to \R_+$ with $\int_S p(\sigma) \dd \sigma =1$. Assume that $|\Lambda|<\infty$, $\int_S p(\sigma) |\log p(\sigma)|\, \dd \sigma<\infty $, and 
	$$
	  	\rho_0 \sup_{(x,\sigma) \in \Lambda\times S}\int_{\Lambda\times S} \bar f\bigl((x,\sigma),(y,\tau)\bigr) 
	  	p(\tau) \dd \tau \dd y \leq \frac{1}{2\e}. 
	$$
	Then 
	\begin{multline} \label{feinternal} 
		\beta \mathcal F_\Lambda[\rho]
			 = |\Lambda|\Bigl( \rho_0(\log\rho_0 - 1) + \rho_0\int_S p(\sigma) \log p(\sigma) \dd \sigma\Bigr) \\
			 	 - \sum_{n=2}^\infty \frac{\rho_0^n}{n!} \int_{\Lambda^n}\int_{S^n} D_n\bigl((x_1,\sigma_1),\ldots, (x_n,\sigma_n) \bigr) \prod_{i=1}^n p(\sigma_i) \dd \vect x \dd \vect \sigma
	\end{multline}
	with absolutely convergent integral and series. 
\end{cor} 

\noindent If $V$ is translation invariant, then the right-hand side of~\eqref{feinternal} is proportional to the volume, up to boundary errors that become irrelevant in the thermodynamic limit, and the corollary also yields an expression for the thermodynamic limit $\lim \frac{1}{|\Lambda|} 		\beta \mathcal F_\Lambda[\rho]$.

The right-hand side of~\eqref{feinternal} corresponds to the functional from Eq.~(27) in~\cite{onsager1949}, which is the free energy functional derived by Onsager \emph{before} applying additional approximations due to thinness of rods etc.

\begin{remark}
	In \cite{jttu2014}, in order to obtain $2$-connected coefficients for the case of molecules with internal degrees of freedom, we needed to assume rigidity of the molecules so that Lemma~4.1 in ~\cite{jttu2014} about factorization of graph weights holds true. In the present article, as seen in Corollary~\ref{cor:feinternal}, we obtain the $2$-connected coefficients  as well provided we keep the probability density $p(\sigma)$ of shapes as an explicit variable. If instead we look at 
	$$
		f_\Lambda(\rho_0):= \inf_p \frac{1}{|\Lambda|} \mathcal F_\Lambda[\rho_0 \, p],
	$$
	expand the minimizer $p(\sigma;\rho_0)$ in powers of $\rho_0$ and compose with the expansion of $\frac{1}{|\Lambda|} \mathcal F_\Lambda[\rho_0 p]$, we see that the coefficient of $\rho_0^n$ in the expansion of $f_\Lambda(\rho_0)$ is \emph{not} given by $D_n$.
\end{remark}

\appendix

\section{Formal power series and Ruelle's algebraic formalism} \label{ApFormal}

Here we summarize some facts on the formal power series used in this article, and point out the relation with Ruelle's algebraic formalism. We are interested in power series and formal power series of the form
\be \label{eq:formal}
	K(z) = K_0 + \sum_{n=1}^\infty \frac{1}{n!} \int_{\mathbb X^n} K_n(x_1,\ldots,x_n)  z(\dd x_1)\cdots  z(\dd x_n)
\ee
where $(\mathbb X,\mathcal{X})$ is a measurable space $z$ is a measure on $(\mathbb X,\mathcal X)$, and $K_0\in \C$ is a scalar, and $K_n:\mathbb X ^n \to \C$ are measurable maps that are invariant under permutation of the arguments.

In general, for a formal power series, the integrals and the series need not to converge, hence, in analogy with the theory of formal power series of a single variable, we define a formal power series as a sequence $(K_n)_{n\in \N}$ of symmetric functions and downgrade~\eqref{eq:formal} to a mnemonic notation. Standard operations such as sums and products are defined directly as operations on the sequences $(K_n)_{n\in \N_0}$ in such a way that for two sufficiently well convergent power series one obtains the same result.   The sum of two formal power series $K+G$ is the formal series with coefficients $(K_n+G_n)_{n\in \N_0}$, for $\lambda\in \C$ the formal  series $\lambda K$ is the series with coefficients  $(\lambda K_n)_{n\in \N_0}$. Other operations are defined below. The resulting algebra of formal power series is exactly the algebra of symmetric functions introduced by Ruelle~\cite[Chapter 4.4]{ruelle1969book}.  \\

\noindent \emph{Product.}
Let $K,G$ be formal power series, then $KG$ is defined by 
\be \label{eq:formal-product}
	(KG)_n (x_1,\ldots,x_n):=\sum_{\ell=0}^n \sum_{J\subset [n], \#J=\ell}  K_\ell\bigl( (x_j)_{j\in J} \bigr)G_{n-\ell}\bigl( (x_j)_{j\in [n]\setminus J}\bigr).
\ee
The empty set $J=\varnothing$ is explicitly allowed. As an operation on sequences of symmetric functions, this is exactly the convolution in~\cite[Chapter 4.4]{ruelle1969book}.
 It is not difficult to check that the product is commutative and associative. 
  Eq.~\eqref{eq:formal-product} generalizes to products $K^{(1)}\cdots K^{(r)}$ as 
\be \label{eq:formal-higher-product}
	\bigl(K^{(1)}\cdots K^{(r)}\bigr)_n(x_1,\ldots,x_r)
		= \sum_{(V_1,\ldots, V_r)} \prod_{\ell =1}^r K^{(\ell)}_{\#V_\ell}\bigl( (x_j)_{j\in V_\ell}\bigr)
\ee
where the sum runs over ordered partitions $(V_1,\ldots, V_r)$ of $[n]$ into $r$ disjoint parts, with $V_i=\varnothing$ explicitly allowed. 

 The definition~\eqref{eq:formal-product} is motivated by the following computation, which is valid if the power series are absolutely convergent: From
\begin{align*}
	K(z) G(z) & = \Biggl( K_0 + \sum_{m=1}^\infty \frac{1}{m!} \int_{\mathbb X^m} K_m(x_1,\ldots,x_m) z(\dd x_1)\cdots z(\dd x_m) \Biggr)  \\
		& \qquad \qquad \times \Biggl( G_0 + \sum_{\ell=1}^\infty \frac{1}{\ell!} \int_{\mathbb X^\ell} G_\ell(x_1,\ldots,x_\ell)  z(\dd x_1)\cdots z(\dd x_\ell) \Biggr) 
\end{align*} 
we get
\begin{align*}
	K(z) G(z) & = \sum_{n=0}^\infty \frac{1}{n!} 			\int_{\mathbb X^n} \Biggl( \sum_{\substack{0\leq m,\ell\leq n\\ m+\ell =n}} \frac{n!}{m! \ell!}\,  K_m(x_1,\ldots,x_m) G_\ell(y_1,\ldots, y_\ell)\Biggr) z^m(\dd \vect x) z^\ell( \dd \vect y). 
\end{align*} 
The summand for $m=\ell = 0$ should be read as $K_0G_0$. 
The binomial coefficient $\binom{n}{m}$ is equal to the number of subsets $J\subset[n]$ of cardinality $\#J =m$. 
The value of the integral 
$$
		\int_{\mathbb X^n} K_m\bigl( (x_j)_{j\in J}\bigr) G_\ell\bigl( (x_j)_{j\in [n]\setminus J}\bigr) z(\dd x_1)\cdots z(\dd x_n)
$$
depends on the cardinality $m$ of $J$ alone, and so we find that 
$$
	K(z) G(z) = \sum_{n=0}^\infty \frac{1} {n!}\int_{\mathbb X^n} (KG)_n(x_1,\ldots, x_n) z(\dd x_1)\cdots z(\dd x_n)
$$
with $(KG)_n$ defined in~\eqref{eq:formal-product}. \\

\noindent {\emph{Variational derivative}}. For $q\in \mathbb X$ and $K$ a formal power series over $\mathbb X$, we define 
\be \label{fvardev}
	\Bigl(\frac{\delta }{\delta z(q)} K\Bigr)_n (x_1,\ldots,x_n) \equiv 	\Bigl(\frac{\delta K}{\delta z}\Bigr)_n (q;x_1,\ldots,x_n)= K_{n+1}(q,x_1,\ldots,x_n). 
\ee
In the language of~\cite[Chapter 4.4]{ruelle1969book},  $\frac{\delta}{\delta z(q)}$ corresponds to the derivation $D_q$. Formally, 
\begin{align*}
	K(z+ t \mu) & = K_0+ \sum_{n=1}^\infty \frac{1}{n!}\int_{\mathbb X^n} K_n(x_1,\ldots, x_n)  \prod_{i=1}^n \bigl( z(\dd x_i)+ t \mu(\dd x_i)\bigr)  \\
	& = K(z) + t \Biggl(\sum_{n=1}^\infty \frac{1}{(n-1)!}
		\int_{\mathbb X^n} K_n(x_1,\ldots, x_n) \mu(\dd x_1) z(\dd x_2)\cdots z(\dd x_n)\Biggr) + O(t^2) \\
		& = K(z) + t \int_{\mathbb X} \Biggl(K_1(q)+   \sum_{n=1}^\infty \frac{1}{n!}
		\int_{\mathbb X^n} K_{n+1}(q,x_1,\ldots, x_n) z(\dd x_1)\cdots z(\dd x_n)\Biggr) \mu(\dd q) + O(t^2)
\end{align*} 
and
\be \label{fvardev2}
	\frac{\dd}{\dd t} K(z+ t\mu)\Bigg|_{t=0} = \int_{\mathbb X} \frac{\delta K}{\delta z}(q;z) \mu (\dd q)
\ee
as it should be. \\

\noindent \emph{Composition I and exponential series.}
Let $F(t) = \sum_{n=0}^\infty f_n t^n/n!$ be a formal power series in a single variable $t$ and $K$ a formal power series on $(\mathbb X,\mathcal X)$ with $K_0=0$. The formal power series $F\circ K$ on $\mathbb X$ is defined by $(F\circ K)_0:=f_0$ and for $n\geq 1$, 
\be \label{composition2} 
	(F\circ K)_n (x_1,\ldots,x_n) := \sum_{m=1}^n\sum_{\{J_1,\ldots,J_m\}\in \mathcal P_n} f_{m} \prod_{\ell=1}^m  K_{\#J_\ell}\bigl( (x_j)_{j \in J_\ell}\bigr)
\ee
with $\mathcal P_n$ the collection of set partitions of $\{1,\ldots,n\}$. Note that only because $K_0=0$ the expression \eqref{composition2} is well-defined, because only in this case the sum is finite. Formally, 
\begin{align*}
	F\bigl( K(z)\bigr) & = f_0 + \sum_{m=1}^\infty \frac{1}{m!} f_m \bigl( K(z)\bigr)^m \\
	&= f_0 + \sum_{m=1}^\infty \frac{1}{m!} f_m \sum_{n=1}^\infty \frac{1}{n!} \int_{\mathbb X^n} \Biggl( \sum_{\substack{(J_1,\ldots, J_m)\\J_1\dot \cup \cdots\dot \cup J_m = [n]}} \prod_{\ell=1}^m  K_{\#J_\ell}\bigl( (x_j)_{j \in J_\ell}\bigr)\Biggr) z^n(\dd \vect x),\\
	&= f_0 +\sum_{n=1}^\infty \frac{1}{n!} f_m  \int_{\mathbb X^n} \Biggl( \sum_{m=1}^\infty \frac{1}{m!}  \sum_{\substack{(J_1,\ldots, J_m)\\J_1\dot \cup \cdots\dot \cup J_m = [n]}} \prod_{\ell=1}^m  K_{\#J_\ell}\bigl( (x_j)_{j \in J_\ell}\bigr)\Biggr) z^n(\dd \vect x)
\end{align*} 
In the second line we have used~\eqref{eq:formal-higher-product}.  Because of $K_0=0$, the only relevant contributions in the last line are from non-empty$J_r$'s. The factor $1/m!$ can be removed if we decide to sum over non-ordered partitions $\{J_1,\ldots, J_m\}$ instead of ordered partitions $(J_1,\ldots, J_r)$, and we arrive at the expression~\eqref{composition2} for the coefficients of $F(K(z))$. 

An important special case is $F(t) = \exp(t)$, for which Eq.~\eqref{composition2} becomes 
\be \label{eq:formal-exponential} 
	(\exp(K))_n (x_1,\ldots,x_n) = \sum_{m=1}^n\sum_{\{J_1,\ldots,J_m\}\in \mathcal P_n}  \prod_{\ell=1}^m  K_{\#J_\ell}\bigl( (x_j)_{j \in J_\ell}\bigr),
\ee
which is exactly the exponential on the algebra of symmetric functions from~\cite[Chapter 4.4]{ruelle1969book}.
\\

\noindent \emph{Composition II}.  
In the proof of Lemma~\ref{lem:treesolution} we need a more general type of composition, namely let $K$ be a formal power series on $\mathbb X$ with $K_0=0$ and $(G(q;z))_{q\in \mathbb X}$ a family of power series 
$$
	G(q;z) =G_0(q)+ \sum_{n=1}^\infty \frac{1}{n!}\int_{\mathbb X^n} G_n(q;x_1,\ldots,x_n) z(\dd x_1) \cdots z(\dd x_n). 
$$
If  $G(q;z)$ is absolutely convergent for each $q$, define 
$$
	\tilde z(\dd q):=G (q;z) z(\dd q), \quad F(z):= K(\tilde z). 
$$
If sums and integrals are absolutely convergent, then 
\begin{align*}
	F(G(\cdot;z)z) \equiv F(z) & = \sum_{m=1}^\infty\frac{1}{m!} \int_{\mathbb X^m} K_m(x_1,\ldots,x_m) G(x_1;z)\cdots G(x_m;z) z(\dd x_1)\cdots z(\dd x_m) \\
	& = \sum_{m=1}^\infty\frac{1}{m!} \int_{\mathbb X^m} K_m(x_1,\ldots,x_m) \\
		&\qquad \times \left( \sum_{r=0}^\infty \frac{1}{r!} \int_{\mathbb X^r} \Biggl( \sum_{\substack{(V_1,\ldots, V_m)\\ V_1\dot \cup \cdots \dot \cup V_m = [r]}} \prod_{\ell=1}^m G_{\#V_\ell} \bigl(x_\ell; (y_j)_{j\in V_\ell}\bigr) \Biggr)z^r (\dd \vect y) \right) z(\dd x_1)\cdots z(\dd x_m), 
\end{align*} 
where the $V_i$ can be empty.
We group pairs $(m,r)$ with a common sum $m+r = n$. For the factorials we note 
$$
	\frac{1}{m!}\frac{1}{r!} = \frac{1}{n!} \binom {n}{m} 
		= \frac{1}{n!}  \#\{J\subset [n]\mid \#J = m\}. 
$$
Exploiting the symmetry of the functions $K_m(\cdot)$ and $G_j(x;\cdot)$, we find that the coefficients of $F$ are given by  
\be \label{fcomp2}
	F_n(x_1,\ldots, x_n)  = \sum_{m=1}^n \sum_{\substack{J\subset [n]  \\ \#J =m}}  K_{m}\bigl((x_j)_{j\in J}\bigr) \sum_{\substack{(V_j)_{j\in J}:\\  \dot \cup_{j\in J} V_j  = [n]\setminus J}} \prod_{j\in J} G_{\# V_j}\bigl(x_j; (x_v)_{v\in V_j}\bigr).
\ee

\section{Holomorphic functions on Banach spaces} \label{app:holo}

Here we collect some fact that are useful for the Banach inversion.  We refer the reader to~\cite{harris2003,mujica2006} for accessible surveys and~\cite{dineen1999book, mujica1986book} for details.  Let $E$ and $F$ be two complex Banach spaces. A  multilinear map $A:E^m\to F$ is bounded if 
$$
	||A||:= \sup \{||A(x_1,\ldots,x_m)||\, \mid x_1,\ldots, x_m\in E,\, \max_{j=1,\ldots,m} ||x_j|| \leq 1\} <\infty.  
$$

\begin{definition}[Homogeneous polynomials and power series] \hfill \label{def:holo1}
	\begin{enumerate}
		\item A mapping $P:E\to F$ is a \emph{continuous $m$-homogeneous polynomial} if there exists a bounded  multilinear map $A:E^m\to F$ such that $P(x) = A(x,\ldots, x)$.
		\item A \emph{power series} from $E$ into $F$ is a series of the form $\sum_{m=0}^\infty P_m (x-a)$, with $a\in E$ and $P_m$ a continuous $m$-homogeneous polynomial. The \emph{radius of uniform convergence} of the series is the supremum over all $r>0$ such that the series converges uniformly on $\{x\in E\mid ||x-a ||\leq r\}$. 
	\end{enumerate} 
\end{definition} 

\noindent The norm of a continuous $m$-homogeneous polynomial $P$ is 
$$
			||P||:= \sup \{ ||Px||\, \mid x \in E:\, ||x||\leq 1 \}. 
$$
For example, if $E=F=\C$ and $P(z) = a_m x^m$, then $||P|| =|a_m|$. 

\begin{prop}[Cauchy-Hadamard formula]\label{prop:roc}
     \cite[Prop. 6]{mujica2006}
	The radius of uniform convergence of the power series $\sum_{m=0}^\infty P_m(x-a)$ satisfies 
	$$
		\frac{1}{R} = \limsup_{m\to \infty} ||P_m||^{1/m}.
	$$
\end{prop} 

\begin{theorem}\cite[Theorem 7]{mujica2006} \label{thm:holo}
	Let $U\subset E$ be a non-empty open subset and $f:U\to F$. The following conditions are equivalent:
	\begin{enumerate}
		\item For each $a\in  U$, the Fr{\'e}chet derivative of $f$ at $a$ exists: i.e., there exists a bounded linear map $A:E\to F$ such that 	
		$$
			|| f(x) - f(a) - A(x-a)|| = o(||x-a||)\qquad (x\to a).
		$$
		\item For each $a\in U$, there exists a power series $\sum_{m=0}^\infty P_m(x-a)$ that converges to $f(x)$ uniformly on some ball $B(a,r)\subset U$ (with $r>0$).
		\item $f$ is continuous in $U$ and, for each $a\in U$, all elements  $\psi$ of the dual Banach space $E'$, and all $b\in E$, the map $\lambda\to \psi( f(a+ \lambda b))$ is holomorphic in the usual sense in the open set $\{\lambda \in \C\mid a+ \lambda b\in U\}$. 
	\end{enumerate} 
\end{theorem} 

\begin{definition} \label{def:holo2}
	A mapping $f:U\to F$ is called \emph{holomorphic} if it satisfies one (hence,  all three) of the conditions (1)-(3) in Theorem~\ref{thm:holo}.
\end{definition} 

\noindent Many theorems for holomorphic functions in $\C$ have analogues (for example, Cauchy integral formulas), but there are a few pitfalls. For example, it is not true that the Taylor series of a function holomorphic on all of $E$ has infinite uniform radius of convergence. Also, it is not true that a holomorphic function is bounded on balls that are bounded away from $\partial U$. 

\begin{example} \cite[Example 2.6]{harris2003}
	Let $c_0(\N)$ be the Banach space of complex-valued sequences that converge to zero, equipped with the usual supremum norm.  Define $f:c_0(\N) \to \C$ by
	$$
		f\bigl( (z_n)_{n\in \N}\bigr):= \sum_{n=1}^\infty z_n^n.
	$$
	Then $f$ is holomorphic on all of $c_0(\N)$, but the radius of uniform convergence (in the sense of Definition~\ref{def:holo1}) of the series is $1$, and for every $r>1$, the function $f$ is unbounded on the ball $\{ z \in c_0(\N) \mid \sup_{n\in \N} |z_n| \leq r\}$.
\end{example} 

\noindent We conclude with a quantitative inverse function theorem. Let $U$ and open subset of $ E$  and $h : U \rightarrow F$. Call $V := h(U)$. An inverse function theorem give condition under which there exist open neighborhoods $U' \subset U$ of $0$ and $V'\subset V$ of $h(0)$, respectively, such that $h:U'\to V'$ is bijection with holomorphic inverse. An quantitative inverse function theorem additionally singles out numbers $r>0$ and $P>0$, which only depends on $U,V, \|Dh(0)^{-1}\|$, for which we may  choose $U' = B_r(0)$ and $V'=h(U')\supset B_P(0)$. Alternatively, one can have quantitative inversion theorem such that  $V'= B_P(0)$ and $U'= h^{-1}(B_P(0))\subset B_r(0)$. Such numbers $r$ and $P$ are sometimes called \emph{Bloch radii} after Bloch's theorem.
 In the following theorem $E=F$. 

\begin{theorem}\cite[Proposition 2]{harris1977}  \label{thm:holoinv}
	Let $B_R(0)$ and $B_M(0)$ be open balls in some complex Banach space $E=F$ and $h : B_R(0) \rightarrow B_M(0)$ a  holomorphic function.  Suppose that  the derivative $\mathrm Dh(0)$ at the origin is invertible with bounded inverse $||\mathrm D h(0)^{-1}||^{-1} \geq a>0$. 
	Let 
	\[
		r=\frac{R^2a}{4M}, \quad P=\frac{R^2 a^2}{8M}.
	\]
	 Then $h$ maps $ B_r(0)$ biholomorphically onto a domain covering $B_P(h(0))$.
\end{theorem}

%\blue{[Optional: cite instead or in addition Theorem~3 in Harris or Theorem 5 in Harris --- check whether these give better bounds]}

%%%%%%%% End stuff 
\subsubsection*{Acknowledgments} 
The main part of this article was completed when the first and third authors were members of the Department of Mathematics
at the University of Sussex and the second was frequently visiting; 
the authors acknowledge the department for the nice atmosphere. 
S. J. thanks the GSSI and T.K. the university in L'Aquila, Italy, for hospitality and M. Lewin for pointing out possible connections with the setting of the Nash-Moser theorem. 

%\nocite{*}
%\bibliographystyle{amsalpha}
%\bibliography{virial}

\providecommand{\bysame}{\leavevmode\hbox to3em{\hrulefill}\thinspace}
\providecommand{\MR}{\relax\ifhmode\unskip\space\fi MR }
\providecommand{\MRhref}[2]{%
  \href{http://www.ams.org/mathscinet-getitem?mr=#1}{#2}
}
\providecommand{\href}[2]{#2}

\end{document}